\documentclass{article}

\usepackage{amsmath}
\usepackage{amsthm}
\usepackage{tikz-cd}

\usepackage[colorlinks=true]{hyperref}
\usepackage{authblk}

\RequirePackage{amsfonts}[1995/01/01]
\DeclareMathSymbol{\rightarrowtail}{\mathrel}{AMSa}{"1A}
\DeclareMathSymbol{\twoheadrightarrow}{\mathrel}{AMSa}{"10}
\DeclareMathSymbol{\nleq}{\mathrel}{AMSb}{"02}
\DeclareMathSymbol{\upharpoonright} {\mathrel}{AMSa}{"16}

\newcommand{\Pred}{{\sf Pred}}
\newcommand{\Rel}{{\sf Rel}}
\newcommand{\Set}{{\sf Set}}

\newcommand{\inl}[1]{\operatorname{inl} #1}
\newcommand{\inr}[1]{\operatorname{inr} #1}
\newcommand{\monic}{\rightarrowtail}
\newcommand{\pow}[1]{\operatorname{\mathcal P}#1}
\newcommand{\trans}[3]{\forcemath{#1 \stackrel{#2}{\rightarrow}{#3}}}
\newcommand{\Trans}[3]{\forcemath{#1 \stackrel{#2}{\Rightarrow}{#3}}}

\newcommand{\Id}{\mbox{\rm Id}}
\newcommand{\id}{\mbox{\rm id}}
\newcommand{\Union}{{\textstyle \bigcup}} 

\newcommand{\sat}[1]{\overline{#1}}
\newcommand{\lax}[1]{\hat{#1}}
\newcommand{\inner}[1]{\check{#1}}
\newcommand{\Atrans}{\mbox{$A${\sf{}-TS}}}
\newcommand{\Ltrans}{\mbox{$(L{+}\tau)${\sf{}-TS}}}
\newcommand{\Lsatts}{\mbox{$L${\sf{}-Sat-TS}}}

\newcommand{\forcemath}[1]
{\relax\ifmmode\expandafter\pione\else\expandafter\pitwo\fi{#1}{$#1$}}
\newcommand{\pione}[2]{#1}
\newcommand{\pitwo}[2]{#2}

\newcommand{\str}{^{\ast}}

\let\temp\phi
\let\phi\varphi
\let\varphi\temp

\newcommand\restr[2]{{
		#1 
		\mathord{\upharpoonright}_{#2} 
}}

\usetikzlibrary{cd,arrows}
\tikzset{
	bend angle=45,
	auto,
	baseline=(current  bounding  box.center),
	>=stealth',
}

\newcommand{\bsat}[1]{\overline{#1}^b} 
\newcommand{\sbsat}[1]{\widetilde{#1}} 
\newcommand{\Giry}{\Pi} 
\newcommand{\Meas}{\sf Meas}
\newcommand{\SRC}[1]{\Sigma_{R}(#1)}
\newcommand{\compl}[1]{#1^{\mathsf{c}}}
\newcommand{\proj}[1]{\pi_{#1}}
\newcommand{\invproj}[1]{\proj{#1}^{-1}}
\newcommand{\Op}[1]{{#1}^\circ}
\newcommand{\N}{\mathbb N}
\newcommand{\RSigma}{R_\Sigma}
\newcommand{\SigmaR}{\Sigma_R}
\newcommand{\tdouble}{t^{*}}
\newcommand{\tsingle}{t^{+}}
\newcommand{\Rdouble}{R^{*}}
\newcommand{\Rsingle}{R^{+}}

\newtheorem{therm}{Theorem}
\newtheorem{corollary}[therm]{Corollary}
\newtheorem{definition}[therm]{Definition}
\newtheorem{example}[therm]{Example}
\newtheorem{lemma}[therm]{Lemma}
\newtheorem{proposition}[therm]{Proposition}
\newtheorem{remark}[therm]{Remark}
\newtheorem{theorem1}[therm]{Theorem}

\begin{document}

%



\title {Bisimulation as a logical relation}

\newcommand{\email}[1]{\href{mailto://#1}{#1}}

\author{Claudio Hermida}
  \affil{School of Computer Science, University of Birmingham  \email{claudio.hermida@gmail.com}}
  \author{Uday Reddy}
  \affil{School of Computer Science, University of Birmingham  \email{u.s.reddy@bham.ac.uk}}
 \author{Edmund Robinson}
  \affil{Electronic Engineering and Computer Science, Queen~Mary University of London \email{e.p.robinson@qmul.ac.uk}}
 \author{Alessio Santamaria}
  \affil{University of Pisa  \email{a.santamaria@qmul.ac.uk}}




\maketitle

\begin{abstract}
	We investigate how various forms of bisimulation can be characterised 
	using the technology of logical relations. The approach taken is that 
	each form of bisimulation corresponds to an algebraic structure 
	derived from a transition system, and the general result is that a 
	relation $R$ between two transition systems on state spaces $S$ 
	and $T$ is a bisimulation if and only if the derived algebraic 
	structures are in the logical relation automatically generated from 
	$R$. We show that this approach works for the original Park-Milner 
	bisimulation and that it extends to weak bisimulation, and branching 
	and semi-branching bisimulation. The paper concludes with a discussion
	of probabilistic bisimulation, where the situation is slightly more 
	complex, partly owing to the need to encompass bisimulations that 
	are not just relations. 
\end{abstract}

\newenvironment{keywords}{\paragraph*{Keywords:}}{}



\section{Introduction}\label{sec:introduction}

This paper is dedicated to John Power, long-time friend and collaborator 
of the authors, whose work in abstract algebra, for example \cite{anderson1997representable}, is guided by a concern for 
practicality led by an understanding of abstract structures that we can 
only aspire to. 

This work forms part of a programme to view logical relations as a structure that arises naturally 
from interpretations of logic and type theory and to expose the possibility of their use as a 
wide-ranging framework for 
formalising links between instances of mathematical structures. See 
\cite{DBLP:journals/entcs/HermidaRR14} for an introduction to this. 
The purpose of this paper is to show 
how several notions of bisimulation (strong, weak, branching and probabilistic) can be viewed as instances of the use of logical relations. 
It is not to prove new facts in process algebra. Indeed the work we produce is based on concrete 
facts, particularly about weak bisimulation, that have long been known in the process algebra 
community. What we do is look at them in a slightly different light. 

Our work is also related to that of the coalgebra community, but is, 
we believe, quite different in emphasis. The main thrust of the related 
work there has been on algebraic theories as formalised by monads. In 
particular, there are abstract notions of bisimulation given in terms of 
monads and monad liftings. This is a presentation-free approach, which has 
both advantages and disadvantages. In this paper, though, we are focusing 
more on presentations of theories and concrete constructions of models. 
The difference is between presenting a group 
structure as an algebra for the group monad, and presenting it directly in 
terms of operations and constants: multiplication, inverse and identity. 
There is a natural notion of congruence between algebras for this 
approach, and it is given by logical relations. 

The primary thrust of this paper is to test the idea that a presentation 
of what is in general a many-sorted mathematical structure, given by types 
and operations, should give a natural notion of congruence between models. We call this the logical relations approach. 
Our tests consist of looking at some of the larger inhabitants of the zoo 
of bisimulations produced by the process algebra community. We will  
show that a number of different notions of bisimulation can be seen as 
the congruences coming from different ways of modelling state transition 
systems. This area has also been studied by the coalgebra community, and 
there are relations between their work and ours that we shall discuss 
later. 

We see there as being advantages in this. A key one is that the concept of 
bisimulation is incorporated as a formal instance of a framework that also 
includes other traditional mathematical 
structure, such as group homomorphisms. 

Formally speaking, the theory of groups is standardly presented as an algebraic theory with  
operations of multiplication ($.$), inverse ($(\ )^{-1}$) and a constant ($e$) giving the identity of the 
multiplication operation. A group is a set equipped with interpretations of these operations under 
which they satisfy certain equations. We will not need to bother with the equations here. If $G$ and 
$H$ are groups, then a group homomorphism $\theta \colon G \longrightarrow H$ is a function $G \longrightarrow H$ between the 
underlying sets that respects the group operations. We will consider the graph of this function as a 
relation between $G$ and $H$. We abuse notation to conflate the function with its graph, and write 
$\theta\subseteq G\times H$ for the relation $(g,\theta g)$.  
Logical relations give a formal way of extending relations to higher types. In particular, the type for 
multiplication is $[(X\times X)\to X]$, and the recipe for $[(\theta\times \theta)\to \theta]$ tells us that 
$(._G, ._H) \in [(\theta\times \theta)\to \theta]$ if and only if for all $g_1,g_2\in G$ and $h_1,h_2\in H$, 
if $(g_1,h_1)\in\theta$ and $(g_2,h_2)\in\theta$, then $(g_1._G g_2, h_1._H h_2)\in\theta$. 
Rewriting this back into the standard functional style, this says precisely that 
$\theta (g_1._G g_2) = (\theta g_1)._H (\theta g_2)$, the part of the standard requirements for a 
group homomorphism relating to multiplication. In other words, this tells us that a relation $\theta$ 
is a group homomorphism between $G$ and $H$ if and only if the operations are in the appropriate 
logical relations for their types and $\theta$ is functional: 
\begin{itemize}
	\item $(._G, ._H) \in [(\theta\times \theta)\to \theta]$
	\item $((\ )^{-1(G)}, (\ )^{-1(H)}) \in [\theta\to \theta]$
	\item $(e_G,e_H)\in\theta$, and 
	\item $\theta$ is functional and total. 
\end{itemize}

We get an equivalent characterisation of (strong) bisimulation. We can take a labelled transition 
system (with labels $A$ and state space $S$) to be an operation of type $[(A\times S) \to \pow S]$, 
or equivalently $[A\to [S\to \pow S]]$. Let $F$ and $G$ be two such (with the same set of labels, but 
state spaces $S$ and $T$), then we show that $R\subseteq S\times T$ is a bisimulation if and only if
the transition operations are in the appropriate logical relation: 
\begin{itemize}
	\item $(F,G) \in [(A\times R) \to \pow R]$, or equivalently
	\item $(F,G) \in [A\to [R \to \pow R]].$
\end{itemize}
Since {\Rel} is a cartesian closed category it does not matter which of these presentations we use, 
the requirement on $R$ will be the same. 

In order to do this we need to account for the interpretation of $\pow$ on relations and this leads us 
into a slightly more general discussion of monadic types. This includes some results about monads 
on {\Set} that we believe are new, or at least are not widely known. 

Weak and branching bisimulation can be made to follow. These forms of 
bisimulation arise in order to deal with the extension of transition 
systems to include silent $\tau$ actions. It is widely known that weak 
bisimulation can be reduced to 
the strong bisimulation of related systems, and we follow this approach. The interest for us is the 
algebraic nature of the construction of the related system, and we give two such, one of which 
explicitly includes $\tau$ actions and the other does not. In this case we get results of the form:
$R\subseteq S\times T$ is a weak bisimulation if and only if
the derived transition operations $\sat F$ and $\sat G$ are in the appropriate logical relation: 
\begin{itemize}
	\item $(\sat F,\sat G) \in [A\to [R \to \pow R]].$
\end{itemize}
This seems something of a cheat but there is an issue here. The $\tau$ actions form a formal part of the semantic 
structure, but are not supposed to be visible. You can argue that is also cheating, and that you would 
really like a semantic structure that does not include mention of $\tau$, and that is what our 
second construction does. 

Branching and semi-branching bisimulations were introduced to deal with 
perceived deficiencies in weak bisimulation. We show that they arise 
naturally out of a variant of the notion of transition system in which 
the system moves first by internal computations to a synchronisation 
point, and then by the appropriate action to a new state. 

Bisimulations between probabilistic systems are a little more problematic. 
They do not quite fit the paradigm because, in the continuous case, we 
have a Markov kernel rather than transitions between particular states. 
Secondly, there are different approaches to bisimilarity. We investigate 
these and show that the logical relations approach can still be extended 
to this setting, and that when we do so there are strong links with these
approaches to bisimilarity. 

The notion of \emph{probabilistic bisimulation} for discrete probabilistic systems is due originally to \cite{larsen_bisimulation_1991}, with further work in \cite{van_glabbeek_reactive_1995}. The continuous case was instead discussed first in \cite{desharnais_bisimulation_2002}, where bisimulation is described as a span of \emph{zig-zag morphisms} between probabilistic transition systems, there called \emph{labelled Markov processes} (LMP), whose set of states is an analytic space. The hypothesis of analyticity is sufficient in order to prove that bisimilarity is a transitive relation, hence an equivalence relation. 
In~\cite{panangaden2009labelled}, the author defined instead the notion of probabilistic bisimulation on a LMP (again with an analytic space of states) as an equivalence relation satisfying a property similar to Larsen and Skou's discrete case. For two LMPs with different sets of states, $S$ and $S'$ say, one can consider equivalence relations on $S+S'$. 

Here we follow the modus operandi of~\cite{de_vink_bisimulation_1999}, where they 
showed the connections between Larsen and Skou's definition in the discrete case and 
the \emph{coalgebraic} approach of the ``transition-systems-as-coalgebras paradigm'' 
described at length in~\cite{rutten_universal_2000}; then they used the same approach 
to give a notion of probabilistic bisimulation in the continuous case of transition systems 
whose set of states constitutes an ultrametric space. In this paper we see LMPs as 
coalgebras for the Giry functor $\Giry \colon \Meas \to \Meas$ (hence we consider 
arbitrary measurable spaces) and a probabilistic bisimulation is defined as a $\Giry$-
\emph{bisimulation}: a span in the category of $\Giry$-coalgebras. At the same time, 
we define a notion of logical relation for two such coalgebras $F \colon S \longrightarrow \Giry S$ 
and $G \colon T \longrightarrow \Giry T$ as a relation $R \subseteq S \times T$ such that 
$(F,G) \in [R \to \Giry R]$, for an appropriately defined relation $\Giry R$. 
It is easy to see that if 
$S=T$ and if $R$ is an equivalence relation, then the definitions of logical relation and 
bisimulation of~\cite{panangaden2009labelled} coincide. What is not straightforward is 
the connection between the definition of $\Giry$-bisimulation and of logical relation in 
the general case: here we present some sufficient conditions for them to coincide, 
obtaining a similar result to de Vink and Rutten, albeit the set of states are not 
necessarily ultrametric spaces.

A second benefit of this approach using explicit algebraic constructions of models is that placing these constructions in this context opens up the possibility of 
applying them in more general settings than {\Set}, by generalising the constructions to 
other frameworks. The early work of 
\cite{hermida1993fibrations,hermida1999some} shows that logical predicates can be obtained from 
quite general interpretations of logic, and more recent work of the authors of this paper shows how 
to extend this to general logical relations. The interpretation of covariant powerset given here is via 
an algebraic theory of complete sup-lattices opening up the possibility of also extending it to more 
general settings (though there will be design decisions about the indexing structures allowed). The 
derived structures used to model weak bisimulation are defined through reflections, and so can be 
interpreted in categories with the correct formal properties. All of this gives, we hope, a framework 
that can be used flexibly in a wide range of settings, see {e.g.} \cite{ghani2010fibrational}.


As we have indicated, much of this is based on material well-known to the process algebra 
community. We will not attempt to give a full survey of sources here. 


\subsection{Related work}

The idea that bisimulation is related to more general notions goes back a 
long way: at least to Aczel's theory of non-well-founded sets 
\cite{aczel1988non}, see also \cite{rutten1992processes}. More recently 
the coalgebra community has engaged heavily with this, both in terms of 
abstracting the notion to general coalgebras and working on abstractions 
of weak bisimulation and, quite recently, branching bisimulation, along 
with forms of probabilistic bisimulation. 

Most of these are based on the notion of transition system as coalgebra 
for a functor that effectively gives the set of possible endpoints for a 
transition starting at a given input state. If this functor is suitably 
well-behaved, or has the right additional structure, then we can get an 
abstract version of, say, weak bisimulation. 

In the specific case of weak 
bisimulation, the basic idea is often to construct the saturation of a 
transition system with $\tau$ moves and to use strong bisimulation on the 
result. This idea dates back a long time to the process algebra community 
around Milner and has to be carried out carefully because expressed as 
simply as above it will yield the wrong results. This is the basic idea 
behind the work of for example, \cite{brengos2015weak}, or 
\cite{sokolova2009coalgebraic}, though in both cases the authors extend 
the idea significantly. Brengos shows that it can be made to carry 
through in a very abstract setting (when the coalgebras on a given object 
are partially ordered and the saturated ones form a reflexive subcategory 
of that partial order). Similarly much of the content of  
\cite{sokolova2009coalgebraic} is that the same abstract approach yields 
a standard form of bisimulation for certain probabilistic systems. 

We have not, however, found work that compares with our characterisation 
in terms of lax transition systems. In fact we suggest that 
this approach departs from ones natural for the coalgebra 
community. If $F$ is a strong monad on a cartesian closed category 
$\sf {C}$, then the internal hom $[c\to Fc]$ is a monoid in 
$\sf {C}$. We can view an $a$-labelled transition system as either a 
morphism $a \longrightarrow [c\to Fc]$, or as a monoid homomorphism 
$a\str \longrightarrow [c\to Fc]$, 
where $a\str$ is the free monoid on $a$. 
We use this formulation to define the notion of lax transition system. 

Some very recent independent work on branching bisimulation deserves 
mention. \cite{beohar2017path} uses a fairly similar approach to us, but is more 
abstract and less specific about synchronisation points. 
\cite{jacobs2021relating} adopts a completely different approach using 
apartness. 

In section~\ref{sec:digressionOnMonads}, our digression on monads, we have a short discussion of 
lifting functors to $\Pred$ and to $\Rel$. There is a considerable body 
of work in this area, some quite general and abstract (including \cite{hermida1998structural}), and we cannot cover 
the relationships with other work in full detail. This kind of area is 
central for the coalgebra community, but we are generally working with 
specific examples, while they are concerned with the abstract properties 
that make arguments go through. Much of the extant work in the area makes 
use of some form of image factorisation in order to get round the issue 
that if $R\subseteq A\times B$ is a relation between $A$ and $B$, and $M$ 
is a functor, then $MR$ has a canonical map to $MA\times MB$, but that map 
is not necessarily monic. Examples include the early work of \cite{hesselink2000fixpoint}, and the foundational work of 
\cite{goubault2008logical}. There is a nice review in 
\cite{kurz2016relation}.  There is also interesting work that 
employs different 
techniques: \cite{sprunger2018fibrational} employs a Kan extension technique, \cite{katsumata2015codensity} uses a double orthogonality technique to induce closure, \cite{baldan2014behavioral} uses quantale-valued relations. \cite{hasuo2013coinductive} uses closure under $\omega$-sequences, a term closure, to induce liftings. Researchers have developed 
the basic image factorisation idea in other directions, for example to 
handle ``up to" techniques, \cite{bonchi2018up}. 

The authors would like to thank Matthew Hennessy for suggesting that weak bisimulation would be 
a reasonable challenge for assessing the strength of this technology, the 
referees of an earlier version for pointing us at branching bisimulation 
as a test case, and referees of this version for helpful suggestions and in particular pressing us to improve 
the situation of the paper with respect to other work. 

\section{Bisimulation}
\label{sec:bisimulation}

The notion of bisimulation was introduced for automata in \cite{park1981concurrency}, 
extended by Milner to processes and then further modified to allow internal actions of those processes,  
\cite{milner1989communication}. The classical notion is {\em strong bisimulation}, defined as a relation 
between labelled transition systems. 

\begin{definition}
	A {\em transition system} consists of a set $S$, together with a function $f:S\longrightarrow \pow S$. We view elements $s\in S$ as states of the system, and read $f(s)$ as the set of states to which $s$ can evolve in a single step. A {\em labelled transition system} consists of a set $A$ of labels (or actions), a set $S$ of states, and a function $F: A \longrightarrow [S\to \pow S]$. For $a\in A$ and $s\in S$ we read $Fas$ as the set of states to which $s$ can evolve in a single step by performing action $a$. $s'\in Fas$ is usually written as $\trans sa{s'}$, 
	using different arrows to represent different $F$'s.
\end{definition}

This definition characterises a labelled transition system as a function from labels to unlabelled transition systems. For each label we get the transition system of actions with that label. By uncurrying $F$ we get an equivalent definition as a function $A\times S  \longrightarrow \pow S$. 

We can now define bisimulation. 

\begin{definition}
	Let $S$ and $T$ be labelled transition systems for the same set of labels, $A$. Then a relation $R\subseteq S\times T$ is a {\em strong bisimulation} if and only if for all $a\in A$, whenever $sRt$
	\begin{itemize}
		\item[-] for all $\trans sa{s'}$, there is $t'$ such that $\trans ta{t'}$ and $s'Rt'$ 
		\item[-] and for all $\trans ta{t'}$, there is $s'$ such that $\trans sa{s'}$ and $s'Rt'$.
	\end{itemize}
\end{definition}

\section{Logical Relations}
\label{sec:logical-relations}

The idea behind logical relations is to take relations on base types, and extend them  to relations on higher types in a structured way. The relations usually considered are binary, but they do not have to be. Even the apparently simple unary logical relations (logical predicates) are a useful tool. In this paper we will be considering binary relations except for a few throwaway remarks. We will also keep things simple by just working with sets. 

As an example, suppose we have a relation $R_0\subseteq S_0 \times T_0$ and a relation $R_1\subseteq S_1 \times T_1$, then we can construct a relation $[R_0\rightarrow R_1]$ between the function spaces $[S_0\rightarrow S_1]$ and $[T_0\rightarrow T_1]$. If $f:S_0\longrightarrow S_1$ and 
$g:T_0\longrightarrow T_1$, then $f [R_0\rightarrow R_1] g$ if and only if for all $s$, $t$ such that $s R_0 t$, then $f(s) R_1 g(t)$. 

The significance of this definition for us is that it arises naturally out of a broader view of the structure. We consider categories of predicates and relations. 

\begin{definition}\label{def:Pred-1}
	The objects of the category {\Pred} are pairs $(P,A)$ where $A$ is a set and $P$ is a subset of $A$. A  
	morphism $(P,A) \longrightarrow (Q,B)$ is a function $f \colon A \longrightarrow B$ such that 
	$\forall a\in A. a\in P \implies f(a) \in Q$. Identities and composition are inherited from \Set. 
\end{definition}

{\Pred} also has a logical reading. We can take $(P,A)$ as a predicate on the type $A$, and 
associate it 
with a judgement of the form $a:A \vdash P(a)$ (read ``in the context $a:A$, $P(a)$ is a proposition''). 
A morphism 
$t \colon (a:A\vdash P(a)) \to (b:B\vdash Q(b))$ 
has two parts: a substitution $b\mapsto t(a)$, and the logical consequence $P(a) \Rightarrow Q(t(a))$ 
(read ``whenever $P(a)$ holds, then so does $Q(t(a))$''). 

\begin{definition}\label{def:Rel-1}
	The objects of the category {\Rel} are triples $(R,A_1,A_2)$ where $A_1$ and $A_2$ are sets and 
	$R$ is a subset of $A_1\times A_2$ (a relation between $A_1$ and $A_2$). A  morphism 
	$(R,A_1,A_2) \longrightarrow (S,B_1,B_2)$ is a pair of functions $f_1 \colon A_1 \longrightarrow B_1$ 
	and $f_2\colon  A_2 \longrightarrow B_2$ such that 
	$\forall a_1\in A_1, a_2\in A_2. (a_1,a_2)\in R \implies (f_1(a_1),f_2(a_2)) \in S$. 
	Identities and composition are inherited from $\Set\times\Set$. 
\end{definition}
\begin{center}
	\begin{tikzcd}
		P  \ar[r, dashed] \ar[d] & Q\ar[d] & R \ar[r, dashed]\ar[d] & S\ar[d]\\
		A \ar[r, "f"] & B   & A_1\times A_2 \ar[r, "f_1\times f_2"] & B_1\times B_2 
	\end{tikzcd}
\end{center}
$\Rel_n$ is the obvious generalisation of $\Rel$ to n-ary relations. 

\Pred{} has a forgetful functor $p\colon \Pred \longrightarrow\Set$, $p(P,A) = A$, and similarly 
{\Rel} has a forgetful functor 
$q\colon \Rel\longrightarrow\Set\times\Set$, $q(R,A_1,A_2) = (A_1,A_2)$, giving rise to two projection functors $\pi_0$ and $\pi_1$ $\Rel\longrightarrow\Set$. These functors carry a good deal of 
structure and are critical to a deeper understanding of the constructions. 

Moreover, both {\Pred} and {\Rel} are cartesian closed categories. 

\begin{lemma} {\Pred} is cartesian closed and the forgetful functor $p:\Pred\to\Set$ preserves that structure. {\Rel} is also cartesian closed and the two projection functors $\pi_0$ and $\pi_1$ preserve that structure. Moreover the function space in {\Rel} is given as in the example above. 
\end{lemma}

So the definition we gave above to extend relations to function spaces can be motivated as the description of the function space in a category of relations. 

\section{Covariant Powerset}
\label{sec:covariant-powerset}

We can do similar things with other type constructions. In particular we can extend relations to relations between powersets. 

\begin{definition}\label{def-collection-rel}
	Let $R\subseteq S\times T$ be a relation between sets $S$ and $T$. We define $\pow R\subseteq \pow S\times \pow T$ by: \\
	$ U [\pow R] V $  if and only if
	\begin{itemize}
		\item[-] for all $u\in U$, there is a $v\in V$ such that $uRv$
		\item[-] and for all $v\in V$, there is a $u\in U$ such that $uRv$
	\end{itemize}
\end{definition}

Again this arises naturally out of the lifting of a construction on {\Set} to a construction on {\Rel}. In this case we have the covariant powerset monad, in which the unit $\eta: S \longrightarrow \pow S$ is $\eta s = \{ s\}$, and the multiplication $\mu: \pow{}^2 S \longrightarrow \pow S$ is $\mu X = \Union X$. 

There are two ways to motivate the definition we have just given. They both arise out of constructions for general monads, and in the case of monads on {\Set} they coincide. 

In {\Pred} our powerset operator sends $(Q,A)$ to $(\pow{Q}, \pow A)$ with the obvious inclusion. 
In {\Rel} it almost  sends $(R, A_1, A_2)$ to 
$(\pow R,\pow{A_1}, \pow{A_2})$, where the ``relation" is as follows: if $U\subseteq R$ ({i.e.} 
$U\in\pow R$) then $U$ projects onto $\mathop{\pi_1} U$ and $\mathop{\pi_2} U$. So for example, if 
$R$ is the total relation on $\{0,1,2\}$ and $U=\{(0,1),(1,2)\}$, then $U$ projects onto $\{0,1\}$ and 
$\{1,2\}$.
The issue is that there are other subsets that project onto the same elements, {e.g.} 
$U'=\{(0,1),(1,1),(1,2)\}$, and hence this association does not give a monomorphic embedding of 
$\pow R$ into 
$\pow{A_1}\times \pow{A_2}$.

\begin{lemma}\label{lemma:egli}
	If $R$ is a relation between sets $A_1$ and $A_2$,  $P_1\subseteq A_1$ and $P_2\subseteq A_2$, 
	then the following are equivalent: 
	\begin{enumerate}
		\item there is $U\subseteq R$ such that $\mathop{\pi_1} U = P_1$ and $\mathop{\pi_2} U = P_2$
		\item for all $a_1\in P_1$ there is an $a_2\in P_2$ such that $a_1 R a_2$ and for all $a_2\in P_2$ 
		there is an $a_1\in P_1$ such that $a_1 R a_2$.
	\end{enumerate} 
\end{lemma}

The latter is the Egli-Milner condition arising in the ordering on the Plotkin powerdomain, 
\cite{plotkin1976powerdomain}.

Thus for {\Rel} we take the powerset of $(R,A_1,A_2)$ to be $(\pow R, \pow{A_1},\pow{A_2})$, where 
$P_1 (\pow R) P_2$ if and only if $P_1$ and $P_2$ satisfy the equivalent conditions of Lemma \ref{lemma:egli}.

\paragraph*{Covariant powerset as the algebraic theory of complete $\vee$-semilattices.}

This form of powerset does not characterise predicates on our starting point. Rather it characterises 
arbitrary  collections of elements of it. To make this precise, consider the following formalisation of the 
theory of complete sup-semilattices. For each set $X$ we have an operation 
$\bigvee_X : L^X \longrightarrow L$. In addition, for any $f:X\longrightarrow Y$, composition with $L^f: L^Y\longrightarrow L^X$ is a substitution that takes an operation of arity $X$ into one of arity $Y$. These operations satisfy the following equations: 
\begin{enumerate}
	\item\label{sup-eq-1} given a surjection $f: X\longrightarrow Y$, $\bigvee_X \circ L^f = \bigvee_Y$. 
	\item\label{sup-eq-2} given an arbitrary function $f: X\longrightarrow Y$, 
	$\bigvee_Y \circ (\lambda {y\in Y}. \bigvee_{f^{-1}\{y\}}\circ L^{i_y}) = \bigvee_X$, where 
	$i_y : f^{-1}\{y\} \longrightarrow X$ is the inclusion of $f^{-1}\{y\}$ in $X$. 
\end{enumerate}

The first axiom generalises idempotence and commutativity of the $\vee$-operator. The second says 
that if we have a collection of sets of elements, take their $\bigvee$'s, and take the $\bigvee$ of the 
results, then we get the same result by taking the union of the collection and taking the $\bigvee$ of 
that. A particular case is that $\bigvee_\emptyset$ is the inclusion of a bottom element. 

The fact that this theory includes a proper class of operators and a proper class of equations does not 
cause significant problems. 

\begin{lemma}
	In the category of sets, $\pow A$ is the free complete sup-semilattice on $A$. 
\end{lemma}

\begin{proof}
	(Sketch) Interpreting the $\bigvee$ operators as unions, it is clear that $\pow A$ is a model of our 
	theory of complete sup-semilattices. 
	
	Suppose now that $f: A \longrightarrow B$ and $B$ is a complete sup-semilattice. Then we have a map 
	$f\str : \pow A \longrightarrow B$ defined by 
	$f\str (X) = \bigvee_X (\lambda x\in X. f(x))$. Equation (\ref{sup-eq-1}) tells us that the operators 
	$\bigvee_X$ are stable under isomorphisms of $X$, and hence we do not need to be concerned 
	about that level of detail. Equation (\ref{sup-eq-2}) now tells us that $f\str$ is a homomorphism. 
	Moreover,  if $X\subseteq A$ then in $\pow A$, $X = \bigvee_X (\lambda x\in X. \{ x \})$. Hence $f\str$ 
	is the only possible homomorphism extending $f$. This gives the free property for $\pow A$. 
\end{proof}

\begin{lemma}
	In {\Pred}, $(\pow P,\pow A)$ is the free complete sup-semilattice on $(P,A)$ and in {\Rel}, 
	$(\pow R, \pow{A_1},\pow{A_2})$ is the free complete sup-semilattice on $(R,A_1,A_2)$. 
\end{lemma}

\begin{proof}
	We start with {\Pred}. For any set $X$, $(X,X)$ is the coproduct in {\Pred} of $X$ copies of $(1,1)$, 
	and $(Q^X,B^X)$ is the product of $X$ copies of $(Q,B)$. $X$-indexed union in the two components 
	gives a map 
	$\bigcup_X : ((\pow{P})^X,(\pow{A})^X) \longrightarrow (\pow P,\pow A)$. 
	Since this works component-wise, these operators satisfy the axioms in the same way as in \Set. 
	$(\pow P,\pow A)$ is thus a complete sup-semilattice. 
	
	Moreover, if $f: (P,A) \longrightarrow (Q,B)$ where $(Q,B)$ is a complete sup-semilattice, then we have 
	$f\str: \pow A \longrightarrow B$ and (the restriction of) $f\str$ also maps $\pow P \longrightarrow Q$. The proof 
	is now essentially as in \Set. 
	
	The proof in {\Rel} is similar. 
\end{proof}

This type constructor has notable differences from a standard powerset. It (obviously) supports 
collecting operations of union, including a form of quantifier: $\bigcup : \pow \pow X \longrightarrow\pow X$. 
However it does not support either intersection or a membership operator. 

\begin{lemma}
	\begin{enumerate}
		\item $\cap : \pow X \times \pow X \to \pow X$ is not parametric.
		\item $\in : X \times \pow X \to 2 = \{\top,\bot\}$ is not parametric. 
	\end{enumerate}
\end{lemma}

\begin{proof}
	Consider sets $A$ and $B$ and a relation $R$ in which $aRb$ and $aRb'$ where $b\neq b'$.
	\begin{enumerate}
		\item $\{a\} \pow R \{b\}$ and $\{a\} \pow R \{b'\}$, but $\{a\}\cap\{a\} = \{a\}$, while 
		$\{b\}\cap\{b'\} = \emptyset$, and it is not the case that $\{a\} \pow R \emptyset$.
		\item  $aRb'$ and $\{a\} \pow R \{b\}$, but applying $\in$ to both left and right components of this gives different results: \\
		$\in (a,\{a\}) = \top$, while $\in (b',\{b\}) = \bot$.
	\end{enumerate}
Hence $\cap$ and $\in$ are not parametric.
\end{proof}

Despite the lack of these operations, this type constructor is useful to model non-determinism. 

\paragraph*{Covariant powerset in {\Rel} using image factorisation.} 

Suppose $Q\subseteq A$, then $\pow Q\subseteq \pow A$, and hence we can easily extend $\pow$ to 
{\Pred}. However, if $R\subseteq A\times B$, then $\pow R$ is a subset of $\pow (A\times B)$, not 
$\pow A \times \pow B$. The consequence is that $\pow$ does not automatically extend to {\Rel} in the same way. 

The second way to get round this is to note that we have projection maps $R \longrightarrow A$ and $R\longrightarrow B$. 
Applying the covariant $\pow$ we get $\pow R \longrightarrow \pow A$ and $\pow R\longrightarrow \pow B$, and hence 
a map 
$\phi: \pow R\longrightarrow (\pow A \times \pow B)$. $\phi$ sends $U\subseteq R$ to 
\[(\pi_A (U), \pi_B (U)) = (\{a\in A\ |\ \exists b\in B.\ (a,b)\in U\},\{b\in B\  |\  \exists a\in A.\ (a,b)\in U\})\]

This map is not necessarily monic:

\begin{example}
	Let $A=\{0,1\}$, $B=\{x,y\}$, and $R=A\times B$. Take $U=\{(0,x),(1,y)\}$, and $V=\{(0,y),(1,x)\}$. Then 
	$\phi U = \phi V = \phi R = A\times B$, and hence $\phi$ is not monic. 
\end{example}

We therefore take its image factorization:
\[\begin{tikzcd}
\pow R \ar[r, two heads] & \overline{\pow R} \ar[r, tail] & \pow A \times \pow B
\end{tikzcd}\]

Using this definition, $\overline{\pow R}$ is 
\[\{ (U,V) \in  \pow A \times \pow B \ |\ \exists S\subseteq R.\ U=\pi_A S \wedge V = \pi_B S \}\]
Now by Lemma \ref{lemma:egli} we have that this gives the same extension of covariant powerset to 
relations as the algebraic approach. 

\begin{lemma}
	The following are equivalent: 
	\begin{enumerate}
		\item $ U [\pow R] V $
		\item there is $S\subseteq R$ such that $\mathop{\pi_A} S = U$ and $\mathop{\pi_B} S = V$
		\item for all $a \in U$ there is an $b\in V$ such that $a R b$ and for all $b \in V$ 
		there is an $a \in U$ such that $a R b$.
	\end{enumerate} 
\end{lemma}

\section{Strong bisimulation via logical relations}

This now gives us the ingredients to introduce the notion of a logical relation between transition systems. 

\begin{definition} 
	Suppose $f : S \longrightarrow \pow S$ and $g: T \longrightarrow \pow T$ are two transition systems. Then we say that $R\subseteq S\times T$ is a {\em logical relation of transition systems} if $(f,g)$ is in the relation $[R\rightarrow \pow R]$.  Similarly, if $A$ is a set of labels and $F: A \longrightarrow [S\rightarrow \pow S]$ and $G: A \longrightarrow [T\rightarrow \pow T]$ are labelled transition systems, then we say that $R\subseteq S\times T$ is a {\em logical relation of labelled transition systems} if $(Fa,Ga)$ is in the relation $[R\rightarrow \pow R]$ for all $a\in A$. 
\end{definition}

The following lemma is trivial to prove, but shows that we could take our uniform approach a step further, to include relations on the alphabet of actions: 

\begin{lemma}
	$R$ is a logical relation of labelled transition systems if and only if $(F,G)$ is in the relation 
	$[\Id_A \rightarrow [R\rightarrow \pow R]]$. 
\end{lemma}

More significantly, we have: 

\begin{lemma}
	If $F: A \longrightarrow [S\rightarrow \pow S]$ and $G: A \longrightarrow [T\rightarrow \pow T]$ are two labelled transition systems, then $R\subseteq S\times T$ is a logical relation of labelled transition systems if and only if it is a strong bisimulation. 
\end{lemma}

\begin{proof}
	The proof is simply to expand the definition of what it means to be a logical relation of labelled transition 
	systems. If $R$ is a logical relation
	and $sRt$ then, applying the definition of logical relation for 
	function space twice, $\{ s' | \trans sa{s'} \} \pow R \{ t' | \trans ta{t'} \}$. So if $\trans sa{s'}$,  then 
	$s' \in \{ s' | \trans sa{s'} \}$. Hence, by definition of $\pow R$ there is a $t' \in \{ t' | \trans ta{t'} \}$ 
	such that $s' R t'$. In other words, $\trans ta{t'}$ and $s'Rt'$. 
	
	Conversely, if $R$ is a strong bisimulation, 
	then $\lambda as.\ \{s' | \trans sa{s'} \}$ and $\lambda at.\ \{t' | \trans ta{t'} \}$ are in the relation 
	$[\Id_A \rightarrow [R\rightarrow \pow R]]$. We have to check that if $a \Id_A a'$ and $sRt$ then 
	$\{s' | \trans sa{s'} \} \pow R \{t' | \trans t{a'}{t'} \}$ But if $a \Id_A a'$, then $a=a'$, so this reduces to 
	$\{s' | \trans sa{s'} \} \pow R \{t' | \trans ta{t'} \}$. Now Definition \ref{def-collection-rel} says that we need 
	to verify that: 
	\begin{itemize}
		\item[-] for all $\trans sa{s'}$, there is $t'$ such that $\trans ta{t'}$ and $s'Rt'$ 
		\item[-] and for all $\trans ta{t'}$, there is $s'$ such that $\trans sa{s'}$ and $s'Rt'$.
	\end{itemize}
	This is precisely the bisimulation condition. 
\end{proof}

This means that we have rediscovered strong bisimulation as the specific notion of congruence for transition systems arising out of a more general theory of congruences between typed structures. 

\section{A digression on Monads}\label{sec:digressionOnMonads}

The covariant powerset functor is an example of a monad, and the two 
approaches given to extend it to {\Rel} at the end of 
section \ref{sec:covariant-powerset} extend to 
general monads. In the case of monads on {\Set} they are equivalent. 

{\Set} satisfies the Axiom Schema of Separation: 
\[ \forall v. \exists w. \forall x. [ x\in w \leftrightarrow x\in v \wedge \phi (x) ] \]
This restricted form of comprehension says that for any predicate $\phi$ on a set $v$, there is a 
subset of $v$ containing exactly the elements of $v$ that satisfy $\phi$. Since this is a set, we can
apply functors to it. 

Moreover, classical sets have the property that any monic whose domain is a non-empty set has 
a retraction. It follows that if $m$ is such a monic, then $Fm$ is also monic, where $F$ is any functor. 

\begin{lemma}\label{lemma:monads-preserve-monics}
	\begin{enumerate}
		\item Let $F \colon \Set \longrightarrow \Set$ be a functor, and $i: A\monic B$ a monic, where $A\neq \emptyset$, then $Fi$ is also monic. 
		\item Let $M:\Set\longrightarrow\Set$ be a monad, and $i:A\monic B$ any monic, then $Mi$ is also monic. 
		\item Let $M:\Set\longrightarrow\Set$ be a monad, then $M$ extends to a functor $\Pred\longrightarrow \Pred$ over \Set.
	\end{enumerate}
\end{lemma}

\begin{proof}
	\begin{enumerate}
		\item $i$ has a retraction which is preserved by $F$. 
		\item If $A$ is non-empty, then this follows from the previous remark. If $A$ is empty, then there are 
		two cases. If $M\emptyset = \emptyset$, then $Mi : \emptyset = M\emptyset = M A \longrightarrow MB$ is
		automatically monic. If $M\emptyset \neq \emptyset$, then let $r$ be any map 
		$B\longrightarrow M\emptyset$. $MB$ is the free $M$-algebra on $B$, and therefore there is a unique 
		$M$-algebra homomorphism $r\str : MB\longrightarrow M\emptyset$ extending this.  $Mi$ is also an $M$-algebra homomorphism and hence so is the composite $r\str (Mi)$. Since $M\emptyset$ is 
		the initial $M$-algebra, it must be the identity, and hence $Mi$ is monic. 
		\item Immediate. \qed
	\end{enumerate}
	\renewcommand{\endproof}{}
\end{proof}

This means that we can make logical predicates work for monads on {\Set}, though there are 
limitations we
will not go into here. We cannot necessarily do the same for monads on arbitrary categories, and we
have already seen that this approach does not work for logical relations. In order to extend to logical 
relations we have our algebraic and image factorisation approaches. 

It is widely known that a large class of monads, monads where the functor preserves filtered (or more 
generally $\alpha$-filtered) colimits correspond to algebraic theories. However it is less commonly 
understood that arbitrary monads can be considered as being given by operations and equations, and 
that the property on the functor is really only used to reduce the collection of operations and equations
down from a proper class to a set. 

Let $M$ be an arbitrary monad on {\Set}, and $\theta: MB \longrightarrow B$ be an $M$-algebra. Let $A$ be an 
arbitrary set, then any element of $MA$ gives rise to an $A$-ary operation on $B$. Specifically, let 
$t$ be an element of $MA$. An $A$-tuple of elements of $B$ is given by a function $e: A\longrightarrow B$, then 
we apply $t$ to $e$ by composing $\theta$ and $Me$ and applying this to $t$: $(\theta\circ (Me))(t)$. 
The monad multiplication can be interpreted as a mechanism for applying terms to terms, and we get 
equations from the functoriality of $M$ and this interpretation of the monad operation. 

We can look at models of this algebraic theory in the category {\Rel} and interpret $MR$ as the free 
model of this theory on $R$. That is the algebraic approach we followed for the covariant powerset $\pow$. 

Alternatively we can follow the second approach and use image factorisation. 
\[
\begin{tikzcd}[sep=large]
M R \ar[r, two heads] 
\ar[d]
& \overline{M R} 
\ar[d,tail]\\
M(A\times B) 
\ar[r, "\langle M\pi_A{,} M\pi_B \rangle"] &
MA\times MB
\end{tikzcd}
\]

Because of the particular properties of {\Set}, monads preserve image factorisation. 

\begin{lemma}\label{lemma:monads-pres-fact}
	Let $M$ be a monad on {\Set}.
	\begin{enumerate}
		\item $M$ preserves surjections: if $f: A\twoheadrightarrow B$ is a surjection from $A$ onto $B$, then $Mf$ is also a surjection. 
		\item $M$ preserves image factorisations: if 
		\begin{tikzcd}[cramped] A \ar[r, twoheadrightarrow,"p"] & P \ar[r,tail,"i"] & B\end{tikzcd} 
		is the image factorisation of $f = i\circ p$, then 
		\begin{tikzcd}[cramped] MA \ar[r, twoheadrightarrow,"Mp"] & MP \ar[r,tail,"Mi"] & MB\end{tikzcd} is the image factorisation of $Mf$. 
	\end{enumerate}
\end{lemma}

\begin{proof}
	\begin{enumerate}
		\item Any surjection in {\Set} is split. The splitting is preserved by functors, and hence surjections are preserved by all functors. 
		\item By Lemma \ref{lemma:monads-preserve-monics}, $M$ preserves both surjections and monics, hence it preserves image factorisations. \qed
	\end{enumerate}
	\renewcommand{\endproof}{}
\end{proof}

Given any monad $M$ 
on {\Set}, $MA\times MB$ is automatically an $M$-algebra with operation 
$\langle \mu_A\circ (M\pi_{MA}), \mu_B\circ (M\pi_{MB}) \rangle: M(MA\times MB) \longrightarrow MA\times MB$. 
Moreover, $\overline{M R}$ is also an $M$-algebra. 

\begin{lemma}
	$\overline{MR}$ is the smallest $M$ sub-algebra of $MA\times MB$ containing the image of $R$. 
\end{lemma}

\begin{proof}
	This follows immediately from the fact that $\overline{M R}$ is an $M$ sub-algebra of 
	$MA\times MB$.
	\[
	\begin{tikzcd}
	M(M R)
	\ar[r, two heads]
	\ar[d, "\mu_R"] &
	M(\overline{M R})
	\ar[r, tail]
	\ar[d, dashed] &
	M(MA\times MB)
	\ar[d, "\langle \mu_A\circ (M\pi_{MA}){,}\mu_B\circ (M\pi_{MB}) \rangle" ] \\
	M R
	\ar[r, two heads] &
	\overline{M R}
	\ar[r, tail] &
	MA\times MB
	\end{tikzcd}
	\]
	In the diagram above, the bottom horizontal composite is $\langle M\pi_A,M\pi_B\rangle$, and the top 
	composite is $M$ applied to this. By Lemma \ref{lemma:monads-pres-fact}, $M$ preserves the image 
	factorization in the bottom composite. It is easy to see that the outer rectangle commutes. It follows 
	that there is a unique map across the centre making both squares commute, and hence that 
	$\overline{MR}$ is an $M$ sub-algebra of $MA\times MB$. 
\end{proof}

The immediate consequence of this is that $\overline{MR}$ is the free $M$ 
algebra on $R$ in {\Rel} and hence the two constructions by free algebra, 
and by direct image coincide in the case of monads on {\Set}.

\section{Monoids}\label{sec:monoids}

Bisimulation is only one of the early characterisations of equivalence for labelled transition systems. 
Another was trace equivalence. That talks overtly about possible sequences 
of actions in a way that bisimulation 
does not. However the sequences are buried in the recursive nature of the 
definition. 

We extend our notion of transition from $A$ to $A\str$, in the usual way. The following is a simple 
induction: 

\begin{lemma}
	If $S$ and $T$ are two labelled transition systems, then $R\subseteq S\times T$ is a bisimulation if and only if for 
	all $w\in A\str$, whenever $sRt$
	\begin{itemize}
		\item[-] for all $\trans sw{s'}$, there is $t'$ such that $\trans tw{t'}$ and $s'Rt'$ 
		\item[-] and for all $\trans tw{t'}$, there is $s'$ such that $\trans sw{s'}$ and $s'Rt'$.
	\end{itemize}
\end{lemma}

In other words, we could have used sequences instead of single actions, and we would have got the 
same notion of bisimulation (but we would have had to work harder to use it). 

Another way of looking at this is to observe that the set of transition systems on $S$, 
$[S\rightarrow \pow S]$, carries a monoid structure. One way of seeing that is to note that 
$[S\rightarrow \pow S]$ is equivalent to the set of $\Union$-preserving endofunctions on $\pow S$. 
Another is that it is the set of endofunctions on $S$ in the Kleisli category for $\pow$. 

More concretely, the unit of the monoid is $\id = \eta = \lambda s. \{s\}$, and the product is got from 
collection, $f_0\cdot f_1 = \lambda s. \Union_{s'\in f_0(s)} f_1(s')$. 

Unsurprisingly, since this structure is essentially obtained from the monad, for any 
$R\subseteq S\times T$, $[R\rightarrow \pow R]$ also carries the structure of a monoid, and the 
projections to $[S\rightarrow \pow S]$ and $[T\rightarrow \pow T]$ are monoid homomorphisms. This 
means that we could characterise strong bisimulations as relations $R$ for which the monoid 
homomorphisms giving the transition systems lift to a monoid homomorphism into the relation. 

\section{Weak bisimulation}

The need for a different form of bisimulation arises when modelling processes. Processes can perform 
internal computations that do not correspond to actions that can be observed directly or synchronised 
with. In essence, 
the state of the system can evolve on its own. This is modelled by incorporating a silent $\tau$ action 
into the set of labels to represent this form of computation. Strong bisimulation is then too restrictive 
because it requires a close correspondence in the structure of the internal computations. 

In order to remedy this, Milner introduced a notion of ``weak'' bisimulation. We follow the account given 
in \cite{milner1989communication}, in which he refers to this notion just as ``bisimulation''. 

We write $A$ (this is Milner's {\it Act}), for the set of possible actions including $\tau$ and $L$ for the actions not including $\tau$. 
So $L = A - \{\tau\}$ and $A = L + \{\tau\}$. If $w \in A\str$, then we write $\hat{w}$ for the sequence 
obtained from $w$ by deleting all occurrences of $\tau$. So $\hat{w}\in L\str$. For example, 
if $w=\tau a_0 a_1 \tau \tau a_0 \tau$, then $\hat{w} = a_0 a_1 a_0$, and if $w' = \tau\tau\tau$, then 
$\hat{w'} = \epsilon$, the empty string. 


\begin{definition}(\cite{milner1989communication})
	Let $S$ be a labelled transition system for $A = L+\{\tau\}$, and $v\in L\str$, then 
	\[ \mbox{$\Trans sv{s'}$ iff there is a $w\in A\str = (L+\{\tau\})\str$ such that $v=\hat{w}$ and $\trans sw{s'}$.} \]
	We can type $\Trans{}{}{}$ as $\Trans{}{}{}: [L\str \to [S\to\pow S]]$, and we refer to it as the 
	{\em system derived from $\trans{}{}{}$.}
\end{definition}

Observe that $\Trans s\epsilon {s'}$ corresponds to 
$\trans s{\tau\str}{s'}$. It follows that $\Trans{}{}{}$ is not quite a 
transition system in the sense previously defined. 
If $S$ is a labelled transition system for $A$, then the extension of $\trans{}{}{}$ to 
$A\str$ gives a monoid homomorphism $A\str \longrightarrow [S\to\pow S]$. However $\Trans{}{}{}$ preserves 
composition but not the identity. We have therefore only a 
semigroup homomorphism 
$L\str \longrightarrow [S\to\pow S]$. This prompts the definition of a lax labelled transition system 
(Definition \ref{def:lax-transition-system}).

We now return to the classical definition of weak bisimulation from \cite{milner1989communication}. 

\begin{definition} 
	If $S$ and $T$ are two labelled transition systems for 
	$A = L+\{\tau\}$, then a relation 
	$R\subseteq S\times T$ is a {\em weak bisimulation} iff for all $a\in A = L+\{\tau\}$, whenever $sRt$
	\begin{itemize}
		\item[-] for all $\trans sa{s'}$, there is $t'$ such that $\Trans t{a}{t'}$ and $s'Rt'$ 
		\item[-] and for all $\trans ta{t'}$, there is $s'$ such that $\Trans s{a}{s'}$ and $s'Rt'$.
	\end{itemize}
\end{definition}

The combination of two different transition relations in this definition is ugly, but fortunately it is well known that we can clean it up by just using the derived relation. 

\begin{lemma}\label{lemma:weak-derived}
	$R$ is a weak bisimulation iff for all $a\in A = L+\{\tau\}$, whenever $sRt$
	\begin{itemize}
		\item[-] for all $\Trans s{\overline{a}}{s'}$, there is $t'$ such that $\Trans t{\overline{a}}{t'}$ and $s'Rt'$ 
		\item[-] and for all $\Trans t{\overline{a}}{t'}$, there is $s'$ such that $\Trans s{\overline{a}}{s'}$ and $s'Rt'$
	\end{itemize}
	where for $x \in L$, $\overline x$ is ``$x$'' seen as a one-letter word, and for $x=\tau$, $\overline x = \epsilon$.
\end{lemma}

We can now extend as before to words in $L\str$. 

\begin{lemma} 
	$R$ is a weak bisimulation iff for all $v\in L\str$, whenever $sRt$
	\begin{itemize}
		\item[-] for all $\Trans sv{s'}$, there is $t'$ such that $\Trans tv{t'}$ and $s'Rt'$ 
		\item[-] and for all $\Trans tv{t'}$, there is $s'$ such that $\Trans sv{s'}$ and $s'Rt'$.
	\end{itemize}
\end{lemma}

Note that we can restrict the underlying alphabet from $A = L+\{\tau\}$ 
to $L$ because $\epsilon\in L\str$ is playing the role of $\tau\in A$.

This now looks very similar to the situation for strong bisimulation. But as we have noted above, 
there is a difference. Previously 
our transition system was given by a monoid homomorphism $A\str \longrightarrow [S\rightarrow \pow S]$. 
Here the identity is not preserved and we only have a homomorphism of semi-groups. 

\begin{lemma}
	If $S$ is a labelled transition system for $A$, then for all $v_0,v_1 \in L\str$, 
	$\Trans {}{v_0 v_1}{} = \  \Trans{}{v_0}{}\cdot \Trans{}{v_1}{}$.
\end{lemma}

In the following sections we present different approaches to understanding weak transition 
systems.

\section{Weak bisimulation through saturation}\label{sec:weak bisim through saturation}

For this section we enrich our setting. For any $S$, $\pow S$ has a natural partial order, and hence so 
do the transition systems on any set $S$, given by the inherited partial order on 
$A \to [S\to \pow S]$. 

\begin{definition}
		 Given transition systems $F: A \longrightarrow [S\to\pow S]$ and $G: A \longrightarrow [T\to\pow T]$, we say that $F\leq G$ 
		iff $S=T$ and $\forall a\in A. \forall s\in S. Fas \leq Gas$. This gives a partial order {\Atrans} that we can 
		view as a category.  
		
		 If $A=L+\{\tau\}$, where $\tau$ is an internal (silent) action, then we shall refer to these as {\em labelled transition systems with internal action} and write the partial order as {\Ltrans}.
\end{definition}

The notion of weak bisimulation applies to transition systems with internal action, while strong 
bisimulation applies to arbitrary transition systems. Our aim is to find a systematic way of deriving the notion of weak bisimulation from strong. 

In the following definition we make use of the fact that $[S\to \pow S]$ is a monoid, as noted in section 
\ref{sec:monoids}.

\begin{definition}
	\label{def:saturated}
	Let $F: (L+\{\tau\}) \longrightarrow [S\to \pow S]$ be a transition system with internal action. We say that $F$ is 
	{\em saturated} if 
	\begin{enumerate}
		\item\label{sat:one} $\id \leq F(\tau)$ and $F(\tau).F(\tau) \leq F(\tau)$ and
		\item\label{sat:two} for all $a\in L$, $F(\tau).F(a).F(\tau) \leq F(a)$
	\end{enumerate}
	We write {\Lsatts} for the full subcategory of saturated transition systems
	with internal actions.
\end{definition}

These conditions are purely algebraic, and so can easily be interpreted in more general settings than 
{\Set}. 

Note that some of the inequalities are, in fact, equalities: 
\[
F(\tau) = F(\tau) . \id \leq F(\tau).F(\tau) \leq F(\tau)
\]
hence $F(\tau).F(\tau) = F(\tau)$. 
Similarly $F(a) = \id.F(a).\id \leq F(\tau).F(a).F(\tau) \leq F(a)$, therefore $F(\tau).F(a).F(\tau)=F(a)$. 

Moreover, if we look at the partial order consisting of unlabelled transition systems on a set $S$, then 
the fact that the monoid multiplication preserves the partial order means that 
$([S\to \pow S], . , \id)$ is a monoidal category. Condition \ref{def:saturated}.\ref{sat:one} says precisely 
that $F(\tau)$ is a monoid in this monoidal category, and condition \ref{def:saturated}.\ref{sat:two} 
that $F(a)$ is an $(F(\tau),F(\tau))$-bimodule. 

The notions of weak and strong bisimulation coincide for saturated transition systems. 

\begin{proposition}
	Suppose $F\colon (L+\{\tau\})\longrightarrow [S\to\pow S]$ and $G \colon (L+\{\tau\})\longrightarrow[T\to\pow T]$ are saturated transition systems with internal actions, then $R\subseteq S\times T$ is a weak bisimulation between the systems if and only if it is a strong bisimulation between them. 
\end{proposition}

\begin{proof}
	In one direction, any strong bisimulation is also a weak one. In the other, suppose $R$ is a weak bisimulation, that $sRt$, and that $\trans{s}{a}{s'}$. Then by definition of weak bisimulation there is 
	$\Trans{t}{a}{t'}$ where $s'Rt'$. We show that $\trans{t}{a}{t'}$. There are two cases: 
	\begin{itemize}
		\item[] $a\neq\tau$: Then, by definition of $\Trans{}a{}$, we have 
		$t (\trans{}\tau{})\str \trans{}a{}  (\trans{}\tau{})\str t'$. But since $F$ is saturated, this implies 
		$\trans{t}{a}{t'}$ as required. 
		\item[] $a=\tau$: Then $t (\trans{}\tau{})\str t'$, and again since $F$ is saturated, this implies 
		$\trans{t}{\tau}{t'}$. 
	\end{itemize}
	Hence we have $\trans{t}{a}{t'}$ and $tRt'$. The symmetric case is identical, so $R$ is a strong 
	bisimulation. 
\end{proof}

Given any transition system with internal action, there is a least saturated transition system containing 
it. 

\begin{proposition}
	The inclusion ${\Lsatts}\hookrightarrow{\Ltrans}$ has a reflection: $\sat{(\cdot)}$. 
\end{proposition}

\begin{proof}
	Suppose $F:(L+\{\tau\}) \longrightarrow [S\to\pow S]$ is a transition system with internal action. 
	Then $F$ is saturated 
	if and only if $F(\tau)$ is a monoid, and $F(a)$ is an $(F(\tau),F(\tau))$-bimodule. 
	So we construct the adjoint by 
	taking $\sat{F}(\tau)$ to be the free monoid on $F(\tau)$ and each $\sat{F}(a)$ to be the free 
	$(\sat{F}(\tau),\sat{F}(\tau))$-bimodule on $F(a)$. This construction works in settings other than {\Set}, but in {\Set} we can give a concrete construction: 
	\begin{itemize}
		\item[] $\sat{F}(\tau) = F(\tau)\str$
		\item[] $\sat{F}(a) = \sat{F}(\tau).F(a).\sat{F}(\tau)$ ($a\neq\tau$) \qed
	\end{itemize}
	\renewcommand{\endproof}{}
\end{proof}

\begin{proposition}
	Suppose $F:(L+\{\tau\}) \longrightarrow [S\to\pow S]$ and $G:(L+\{\tau\})\longrightarrow [T\to\pow T]$ are transition systems 
	with internal actions (not necessarily saturated), then $R\subseteq S\times T$ is a weak bisimulation 
	between $F$ and $G$ if and only if it is a strong bisimulation between $\sat{F}$ and $\sat{G}$. 
\end{proposition}

\begin{proof}
	This is a direct consequence of the concrete construction of the saturated reflection. It follows from 
	Lemma \ref{lemma:weak-derived}, since the transition relation on the saturation is the derived 
	transition relation on the 
	original transition system: $\trans{s}{a}{s'}$ in $\sat{F}$ if and only if $\Trans{s}{\sat{a}}{s'}$ with respect 
	to $F$ (and similarly for $G$). 
\end{proof}

\begin{corollary}
	Suppose $F\colon(L+\{\tau\})\longrightarrow [S\to\pow S]$ and $G\colon(L+\{\tau\})\longrightarrow [T\to\pow T]$ are transition systems 
	with internal actions, and $R\subseteq S\times T$.  Then the following are equivalent: 
	\begin{enumerate}
		\item $R$ is a weak bisimulation between $F$ and $G$ 
		\item $\sat{F}$ and $\sat{G}$ are in the appropriate logical relation: 
		$(\sat{F},\sat{G})\in[\Id_{L+\{\tau\}}\to [R\to\pow R]]$
		\item $R$ is the state space of a saturated transition system in {\Rel} whose first projection is 
		$\sat{F}$ and whose second is $\sat{G}$. 
	\end{enumerate}
\end{corollary}

The consequence of this is that we now have two separate ways of giving semantics to transition 
systems with inner actions. Given $F\colon(L+\tau)\longrightarrow[S\to\pow S]$, we can just take $F$ as a transition 
system. If we then apply the standard logical relations framework to this definition we get that two 
such, $F$ and $G$, are related by the logical relation $[\Id_{(L+\tau)} \to [R\to\pow R]]$ if and only if 
$R$ is a strong bisimulation between $F$ and $G$. If instead we take the semantics to be $\sat{F}$, 
typed as $\sat{F}:(L+\tau)\longrightarrow[S\to\pow S]$, then $\sat{F}$ and $\sat{G}$ are related by the logical 
relation $[\Id_{(L+\tau)} \to [R\to\pow R]]$ if and only if 
$R$ is a weak bisimulation between $F$ and $G$.

\section{Lax transition systems}\label{sec:lax transition systems}

Saturated transition systems still include explicit $\tau$-actions even though these are supposed to be 
internal actions only indirectly observable. We can however avoid $\tau$'s appearing explicitly in the 
semantics by giving a relaxed variant of the monoid semantics. 

We recall that for an arbitrary set of action labels $A$, the set of $A$-labelled transition systems 
$A\longrightarrow[S\to\pow S]$ is isomorphic to the set of monoid homomorphisms $A\str\longrightarrow[S\to\pow S]$, and 
moreover that for any transition systems $F$ and $G$ and relation $R\subseteq S\times T$, 
$F$ is related to $G$ by $[{\Id_A}\to[R\to\pow R]]$ iff $F$ is related to $G$ as 
monoid homomorphism by $[{\Id_{A\str}}\to [R\to\pow R]]$ iff $R$ is a strong bisimulation between $F$ 
and $G$. 

We can model transition systems with internal actions similarly, by saying what transitions correspond 
to sequences of visible actions. The price we pay is that, since $\tau$ is not visible, we have genuine 
state transitions corresponding to the empty sequence. We no longer have a monoid homomorphism. 

\begin{definition}
	\label{def:lax-transition-system}
	A {\em lax transition system} on an alphabet $L$ (not including an internal action $\tau$) is a 
	function $F:L\str \longrightarrow [S\to\pow S]$ such that: 
	\begin{enumerate}
		\item $\id\leq F(\epsilon)$ (reflexivity)
		\item $F(vw) = F(v).F(w)$ (composition)
	\end{enumerate}
\end{definition}

\begin{definition}
	Let $F:(L+\{\tau\})\longrightarrow [S\to\pow S]$ be a transition system with internal action, then its {\em laxification}
	$\lax{F}:L\str \longrightarrow [S\to\pow S]$ is the lax transition system defined by: 
	\begin{enumerate}
		\item $\lax{F} (\epsilon) = F(\tau)\str$
		\item $\lax{F} (a) = F(\tau)\str. F(a). F(\tau)\str$, for any $a\in L$. 
		\item $\lax{F} (vw) = \lax{F} (v) . \lax{F} (w)$. 
	\end{enumerate}
\end{definition}

It is trivial that $\lax{F}$ is a lax transition system. 

\begin{lemma}
	If $F:(L+\{\tau\})\longrightarrow [S\to\pow S]$ is a transition system with internal action, then its laxification
	$\lax{F}:L\str\longrightarrow [S\to\pow S]$ is a lax transition system.
\end{lemma}

We have reproduced the derived transition system. 

Note that if $G$ is a lax transition system, then $G(w)$ depends only on $G(\epsilon)$ and the $G(a)$, 
all other values are determined by composition. Note also that if $F$ is saturated, then 
$\lax{F}(\epsilon) = F(\tau)$ and $\lax{F}(a) = F(a)$. 

We can also go the other way. Given a lax transition system, $F:L\str\longrightarrow [S\to\pow S]$, then we can 
define a transition system with inner action: 
$\inner{F}:(L+\{\tau\})\longrightarrow [S\to\pow S]$ where 
\begin{itemize}
	\item $\inner{F} (\tau) = F(\epsilon)$
	\item $\inner{F} (a) = F(a)$
\end{itemize}

\begin{lemma}
	If $F:(L+\{\tau\})\longrightarrow [S\to\pow S]$ is a transition system with internal action, then its saturation 
	$\sat{F}$ can be constructed as $\inner{\lax{F}}$. 
\end{lemma}

One way of looking at this is that a lax transition system is just a saturated one in thin disguise. But 
from our perspective it gives us a different algebraic semantics for transition systems with inner 
action that can also be made to account for weak bisimulation, and this time the $\tau$ actions do 
not appear in the formal statement. 

\begin{lemma}
	Suppose $F:(L+\{\tau\})\longrightarrow [S\to\pow S]$ and $G:(L+\{\tau\})\longrightarrow [T\to\pow T]$ are transition systems 
	with internal actions, and $R\subseteq S\times T$.  Then the following are equivalent: 
	\begin{enumerate}
		\item $R$ is a weak bisimulation between $F$ and $G$ 
		\item $(\lax{F},\lax{G})\in [\Id_{L\str}\to [R\to\pow R]]$ 
		\item $R$ is the state space of a lax transition system in {\Rel} whose first projection is 
		$\lax{F}$ and whose second is $\lax{G}$. 
	\end{enumerate}
\end{lemma}

\section{(Semi-)Branching bisimulations}

In this section, we shall always consider two labelled transition systems $F \colon (L + \{\tau\}) \longrightarrow [ S \to \pow S]$ and $G \colon (L + \{\tau\}) \longrightarrow [T \to \pow T]$ with an internal action $\tau$. We begin by introducing the following notation: we say that $x \overset{\tau^*}{\to} y$, for $x$ and $y$ in $S$ (or in $T$) if and only if there is a finite, possibly empty, sequence of $\tau$ actions
\[
x \overset \tau \to \cdots \overset \tau \to y;
\]
if the sequence is empty, then we require $x = y$.

We now recall the notion of \emph{branching} bisimulation, which was introduced in~\cite{van_glabbeek_branching_1996}.

\begin{definition}\label{def:branching bisimulation}
	A  relation $R \subseteq S \times T$ is called a \emph{branching bisimulation} if and only if whenever $s R t$:
	\begin{itemize}
		\item $s \overset{a}{\to} {s'}$ implies $\bigl( (\exists t_1,t_2 \in T \ldotp t \overset {\tau^*} \to {t_1} \overset a \to {t_2} \land s R t_1 \land s' R t_2) \text{ or } ( a = \tau \land s' R t) \bigr)$,
		\item $t \overset{a}{\to} {t'}$ implies $\bigl( (\exists s_1,s_2 \in S \ldotp s \overset {\tau^*} \to {s_1} \overset a \to {s_2} \land s_1 R t \land s_2 R t') \text{ or } ( a = \tau \land s R t') \bigr)$.
	\end{itemize}
\end{definition}

\begin{remark}\label{rem:branching tau-case}
	In particular, if $R$ is a branching bisimulation, $s R t$ and $s \overset \tau \to {s'}$ then there exists $t' \in T$ such that $t \overset {\tau^*} \to {t'}$ and $s' R t'$.
\end{remark}

We show how branching bisimulation is also an instance of logical relation between appropriate derived versions of $F$ and $G$.

\begin{definition}\label{def:branching saturation}
	The \emph{branching saturation} of $F$, denoted by $\bsat F$, is a function
	\[
	\bsat F \colon (L+\{\tau\}) \longrightarrow [ S \to \pow {(S \times S)}]
	\]
	defined as follows. Given $s \in S$ and $a \in L + \{\tau\}$, 
	\[
	\bsat F a s = \{ (s_1,s_2) \in S \times S \mid (s \overset {\tau^*} \to s_1 \overset a \to s_2) \text{ or } (a = \tau \text{ and } s=s_1=s_2) \}.
	\]
\end{definition}

\begin{theorem1}\label{thm:branching iff logical relation}
	Let $R \subseteq S \times T$. Then $R$ is a branching bisimulation if and only if $(\bsat F, \bsat G) \in [\Id_{L+\{\tau\}} \to [R \to \pow{(R \times R)}]]$.
\end{theorem1}
\begin{proof}
	Let us unpack the definition of the relation $[\Id_{L+\{\tau\}} \to [R \to \pow{(R \times R)}]]$. We have that $(\bsat F, \bsat G) \in [\Id_{L+\{\tau\}} \to [R \to \pow{(R \times R)}]]$ if and only if for all $a \in L+\{\tau\}$ and for all $s \in S$ and $t \in T$ such that $s R t$ we have $(\bsat F a s) [\pow{(R \times R)}] (\bsat G a t)$. By definition of $\pow{(R \times R)}$, this means that for all $(s_1,s_2) \in \bsat F a s$ there exists $(t_1,t_2)$ in $\bsat G a t$ such that $s_1 R t_1$ and $s_2 R t_2$.
	
	Suppose then that $R$ is a branching bisimulation, consider $s R t$ and take $(s_1,s_2) \in \bsat F a s$. We have two possible cases to discuss: $a = \tau \text{ and } s=s_1=s_2$, or $s \overset {\tau^*} \to s_1 \overset a \to s_2$. In the first case, consider the pair $(t,t)$: this clearly belongs to $\bsat G a t$. In the second case, we are in the following situation:
	\[
	\begin{tikzcd}
	s \ar[r,-,"R"] \ar[d,"\tau^*"'] & t \\
	s_1 \ar[d,"a"'] \\
	s_2
	\end{tikzcd}
	\]
	If $\tau^*$ is the empty list, then $s=s_1$, hence $s_1 R t$: by definition of branching bisimulation, there are indeed $t_1$ and $t_2$ such that:
	\[
	\begin{tikzcd}[row sep=1em]
	& t \ar[dd,"\tau^*"] \\
	s \ar[ur,-,"R"] \ar[dr,-,"R"] \ar[dd,"a"']\\
	& t_1 \ar[d,"a"] \\
	s_2 \ar[r,-,"R"] & t_2
	\end{tikzcd}
	\]
	hence $(t_1,t_2) \in \bsat G a t$. If $\tau^*=\tau^n$, with $n \ge 1$, then by Remark~\ref{rem:branching tau-case} applied to every $\tau$ in the list $\tau^*$, there exists $t'$ in $T$ such that $t \overset{\tau^*}{\to} t'$ and $s_1 R t'$. Now apply again the definition of branching bisimulation for $s R t'$: we have that there are $t_1$ and $t_2$ in $T$ such that:
	\[
	\begin{tikzcd}[row sep=1em]
	& t \ar[dd,"\tau^*"]  \\
	s \ar[ur,-,"R"] \ar[dd,"\tau^*"'] 					  \\
	& t' \ar[dd,"\tau^*"] \\
	s_1 \ar[dr,-,"R"] \ar[dd,"a"'] \ar[ur,-,"R"]						  \\
	& t_1 \ar[d,"a"] 	  \\
	s_2 \ar[r,-,"R"] 								& t_2
	\end{tikzcd}
	\]
	hence $(t_1,t_2)\in \bsat G a t$. This proves that if $R$ is a branching bisimulation, then $(\bsat F, \bsat G) \in [\Id_{L+\{\tau\}} \to [R \to \pow{(R \times R)}]]$.
	
	Conversely, suppose $(\bsat F, \bsat G) \in [\Id_{L+\{\tau\}} \to [R \to \pow{(R \times R)}]]$ and that we are in the following situation:
	
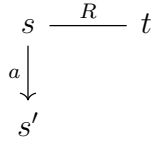
\[
	\begin{tikzcd}
	s \ar[r,-,"R"] \ar[d,"a"'] & t \\
	s'
	\end{tikzcd}
	\]
	Then we have $(s,s') \in \bsat F a s$, because indeed $s \overset {\tau^*} \to s \overset a \to s'$. By definition of the relation $\pow{(R \times R)}$, there exists $(t_1,t_2) \in \bsat G a t$ such that $s R t_1$ and $s' R t_2$. It is immediate to see that this is equivalent to the condition required by Definition~\ref{def:branching bisimulation}, hence $R$ is in fact a branching bisimulation.    
\end{proof}

In~\cite{van_glabbeek_branching_1996} also a weaker notion of branching bisimulation was introduced, which we recall now.

\begin{definition}
	A relation $R \subseteq S \times T$ is called a \emph{semi-branching bisimulation} if and only if whenever $sRt$:
	\begin{itemize}
		\item $s \overset{a}{\to} {s'}$ implies  $\bigl( (\exists t_1,t_2 \in T \ldotp t \overset {\tau^*} \to {t_1} \overset a \to {t_2} \land s R t_1 \land s' R t_2)$ or $( a = \tau \land \exists t' \in T \ldotp {t \overset {\tau^*} \to t'} \land s R t' \land s' R t') \bigr)$,
		\item $t \overset{a}{\to} {t'}$ implies  $\bigl( (\exists s_1,s_2 \in S \ldotp s \overset {\tau^*} \to {s_1} \overset a \to {s_2} \land s_1 R t \land s_2 R t')$ or $( a = \tau \land \exists s' \in S \ldotp {s \overset {\tau^*} \to s'} \land s' R t \land s' R t') \bigr)$.
	\end{itemize}
\end{definition}

Every branching bisimulation is also semi-branching, but the converse is not true. The difference between branching and semi-branching bisimulation is in what is allowed to happen in the $\tau$-case. Indeed, if $s \overset{\tau}{\to} s'$ and $sRt$, in the branching case it must be that either also $s'Rt$, or $t$ can ``evolve'' into $t_1$, for $sRt_1$, by means of zero or more $\tau$ actions, and then $t_1$ has to evolve into a $t_2$ via a $\tau$ action with $s' R t_2$. In the semi-branching case, $t$ is always allowed to evolve into $t'$ with zero or more $\tau$ steps, as long as $s$ is still related to $t'$, as well as $s' R t'$. Figure~\ref{fig:branching vs semibranching} shows this in graphical terms.

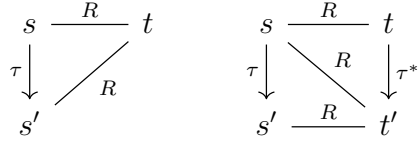
\begin{figure}
	\[
	\begin{tikzcd}
		s \ar[d,"\tau"'] \ar[r,"R",-] & t \\
		s' \ar[ur,-,"R"']
	\end{tikzcd}
	\qquad
	\begin{tikzcd}
		s \ar[d,"\tau"'] \ar[r,"R",-]  \ar[dr,"R",-] & t \ar[d,"\tau^*"] \\
		s' \ar[r,-,"R"] & t'
	\end{tikzcd}
	\]
	\caption{Difference between branching (left) and semi-branching (right) case for $\tau$ actions.}
	\label{fig:branching vs semibranching}
\end{figure}

We can prove a result analogous to Theorem~\ref{thm:branching iff logical relation} for semi-branching bisimulations. To do so, we introduce an appropriate derived version of a labelled transition system $F \colon (L+\{\tau\}) \longrightarrow [S \to \pow{(S\times S)}]$.

\begin{definition}
	The \emph{semi-branching saturation} of $F$, denoted by $\sbsat F$, is a function
	\[
	\sbsat F \colon (L+\{\tau\}) \longrightarrow [ S \to \pow{(S \times S)}]
	\]
	defined as follows. Given $s \in S$ and $a \in L+\{\tau\}$,
	\[
	\sbsat F a s = \{ (s_1,s_2) \in S \times S \mid (s \overset {\tau^*} \to s_1 \overset a \to s_2) \text{ or } (a = \tau \text{ and } s_1 = s_2 \text{ and } s \overset {\tau^*} \to s_1 \}.
	\]
\end{definition}

Notice that Remark~\ref{rem:branching tau-case} continues to hold for semi-branching bisimulations too.

\begin{theorem1}
	Let $R \subseteq S \times T$ be a relation. Then $R$ is a semi-branching bisimulation if and only if $(\sbsat F, \sbsat G) \in [\id_{L+\{\tau\}} \to [R \to \pow{(R \times R)}]]$.
\end{theorem1}
\begin{proof}
	Same argument of the proof of Theorem~\ref{thm:branching iff logical relation}.
\end{proof}

\section{The almost-monad}
In Section~\ref{sec:monoids}, we observed that $[S \to \pow{S}]$ enjoys a monoid structure inherited from the monadicity of the covariant powerset $\pow{}$. Sadly, we cannot say quite the same for $[S \to \pow{(S \times S)}]$. Indeed, consider the functor $T(A)=\pow{(A \times A)}$:
\[
\begin{tikzcd}[arrow style=tikz,row sep=0em,column sep=3em]	
\Set \ar[r,"- \times -"] \ar[rr,bend angle=30,bend left,"T"] & \Set \ar[r,"\pow{}"] & \Set \\
A \ar[r,|->] \ar[d,"f"'] & A \times A \ar[r,|->] \ar[d,"f \times f"] & \pow{(A \times A)} \ar[d,"\pow{(f\times f)}"] \\[2em]
B \ar[r,|->] & B \times B \ar[r,|->] & \pow{(B \times B)}
\end{tikzcd}
\]
where $\pow{(f\times f)}(S) = \{\bigl(f(x),f(y)\bigr) \mid (x,y) \in S\}$. We can define two natural transformations $\eta \colon \Id_\Set \longrightarrow T$ and $\mu \colon T^2 \longrightarrow T$ as follows: $\eta_A (a) = \{(a,a)\}$ and
\[
\begin{tikzcd}[row sep=0em,arrow style=tikz]
\pow{}\bigl( \pow{(A \times A)} \times \pow{(A \times A)}  \bigr) \ar[r,"\mu_A"] & \pow{(A \times A)} \\
U \ar[r,|->] & \bigcup\limits_{(V,W) \in U} (V \cup W)
\end{tikzcd}
\]
It is not difficult to see that $\eta$ and $\mu$ are indeed natural, and that the following square commutes for every set $A$:
\[
\begin{tikzcd}[arrow style=tikz]
T^3 A \ar[r,"\mu_{TA}"] \ar[d,"T\mu_A"'] & T^2 A \ar[d,"\mu_A"] \\
T^2 A \ar[r,"\mu_A"'] & TA
\end{tikzcd}
\]
However, although the left triangle in the following diagram commutes, the right one fails to do so in general:
\[
\begin{tikzcd}[arrow style=tikz]
TA \ar[r,"\eta_{TA}"] \ar[dr,"\id_{TA}"'] & T^2 A \ar[d,"\mu_A"] & TA \ar[l,"T\eta_A"'] \ar[dl,"\id_{TA}"] \\
& TA
\end{tikzcd}
\]
Indeed, given $S \subseteq A \times A$, it is true that $S \cup S = S$, but 
\[
\mu_A\bigl(T\eta_A(S)\bigr) = \mu_A\Bigl( \left\{ \bigl( \{(x,x)\}, \{(y,y)\} \bigr) \mid (x,y) \in S \right\} \Bigr) = \bigcup_{(x,y)\in S} \bigl( \{(x,x)\} \cup \{(y,y)\} \bigr) \ne S.
\]
This means that $(T,\eta,\mu)$ falls short of being a monad: it is only a ``left-semi-monoid'' in the category of endofunctors and natural transformations on $\Set$, in the sense that $\eta$ is only a left unit for the multiplication $\mu$.

One can go further, and build up the ``Kleisli non-category'' associated to $(T,\eta,\mu)$, following the usual definition for Kleisli category of a (proper) monad, where morphisms $A \longrightarrow B$ are functions $A \longrightarrow \pow{(B \times B)}$, and composition of $f \colon A \longrightarrow \pow{(B \times B)}$ and $g \colon B \longrightarrow \pow{(C \times C)}$ is the composite in $\Set$:
\[
\begin{tikzcd}[row sep=0em,arrow style=tikz]
A \ar[r,"f"] & TB \ar[r,"Tg"] & T^2 B \ar[r,"\mu_B"] & TB \\
a \ar[r,|->] & f(a) \ar[r,|->] & \{ (g(x),g(y)) \mid (x,y) \in f(a)  \} \ar[r,|->] & \bigcup\limits_{(x,y)\in f(a)} (g(x) \cup g(y))
\end{tikzcd}
\]
This composition law has $\eta$ as a left-but-not-right identity. Whereas the set of endomorphisms on $A$ in the Kleisli category of a proper monad is always a monoid with the multiplication defined as the composition above, here we get that $[A \to \pow{(A \times A)}]$ is only a left-semi-monoid.

We can define a partial order on $[A \to \pow{(A \times A)}]$ in a canonical way, by setting $f \le g$ if and only if for all $a \in A$ $f(a) \subseteq g(a)$; by doing so, we can regard $[A \to \pow{(A \times A)}]$ as a category. The multiplication $f \cdot g \colon A \longrightarrow \pow{(A \times A)}$, defined as $f \cdot g (a)=\bigcup_{(x,y)\in f(a)} \bigl(g(x) \cup g(y)\bigr)$, preserves the partial order, therefore $[A \to \pow{(A \times A)}]$ is a ``left-semi-monoidal'' category.

\section{Branching and semi-branching saturated systems}

In this section we investigate the properties of $\bsat F(\tau)$ and $\sbsat F(\tau)$ as elements of $[S \to \pow{(S \times S)}]$, for $F \colon (A+\{\tau\}) \longrightarrow [S \to \pow{(S \times S)}]$, to explore whether it is possible to define an appropriate notion of branching or semi-branching saturated systems, where strong and branching (or semi-branching) bisimulations are the same, cf.\ weak case in Sections~\ref{sec:weak bisim through saturation} and~\ref{sec:lax transition systems}.

\begin{lemma}
	$\eta_S \le \bsat F (\tau)$, but $\bsat F(\tau) \cdot \bsat F(\tau) \nleq \bsat F(\tau)$ in general.
\end{lemma}
\begin{proof}
	By definition, the pair $(s,s)$, for $s \in S$, belongs to $\bsat F \tau (s)$, hence $\eta_S \le \bsat F (\tau)$.
	
	Let now $(x,y) \in (\bsat F(\tau) \cdot \bsat F(\tau))(s) = \bigcup_{(s_1,s_2) \in \bsat F \tau (s)} \bigl( \bsat F\tau (s_1) \cup \bsat F \tau(s_2) \bigr)$: we want to check whether $(x,y) \in \bsat F \tau (s)$. Suppose that $(x,y) \in \bsat F \tau (s_1)$ for some $(s_1,s_2) \in \bsat F \tau (s)$. Then we are in one of the following four situations:
	\begin{enumerate}
		\item 
		$
		\begin{tikzcd}
		s \ar[r,"\tau^*"] & s_1 \ar[r,"\tau"] \ar[dr,"\tau^*"] & s_2 \\
		&										& x \ar[r,"\tau"] & y
		\end{tikzcd}
		$
		\item 
		$
		\begin{tikzcd}
		s \ar[r,equal] & s_1 \ar[r,equal] \ar[dr,"\tau^*"] & s_2 \\
		&									& x \ar[r,"\tau"] & y
		\end{tikzcd}
		$
		\item 
		$
		\begin{tikzcd}
		s \ar[r,"\tau^*"] & s_1 \ar[r,"\tau"] \ar[d,equal] & s_2 \\
		& x \ar[r,equal] & y
		\end{tikzcd}
		$
		\item 
		$
		\begin{tikzcd}[row sep=1em,column sep=1em]
		s \ar[r,equal] & s_1 \ar[r,equal] \ar[d,equal] & s_2 \\
		& x \ar[r,equal] & y
		\end{tikzcd}
		$
	\end{enumerate}
	In cases 1 and 2, we can conclude that
	$
	\begin{tikzcd}[cramped,sep=small]
	s \ar[r,"\tau^*"] & x \ar[r,"\tau"] & y
	\end{tikzcd}
	$, while in case 4 we get $s=x=y$, hence $(x,y) \in \bsat F \tau (s)$. However, if in case 3 we are in the situation whereby $s \ne s_1$, then $(x,y)\notin \bsat F \tau (s)$, as it is neither the case that $s=x=y$ nor
	$
	\begin{tikzcd}[cramped,sep=small]
	s \ar[r,"\tau^*"] & x \ar[r,"\tau"] & y.
	\end{tikzcd}
	$
\end{proof}

It turns out, however, that the semi-branching saturation of $F$ behaves much better than $\bsat F$.

\begin{lemma}
	$\sbsat F (\tau)$ is a left-semi-monoid in $[S \to \pow{(S \times S)}]$, and $\sbsat F (a)$ is a left $\sbsat F (\tau)$-module for all $a \in A$.
\end{lemma}
\begin{proof}
	Again, it is immediate to see that $\eta_S \le \sbsat F(\tau)$, because 
	$
	\begin{tikzcd}[cramped,sep=small]
	s \ar[r,"\tau^*"] & s
	\end{tikzcd}
	$
	for any $s$, given that $\tau^*$ can be the empty list of $\tau$'s.
	
	Now we prove that $\sbsat F (\tau) \cdot \sbsat F (\tau) \le \sbsat F(\tau)$. Let $s \in S$ and $(x,y) \in (\sbsat F (\tau) \cdot \sbsat F (\tau))(s)$. Then there exists a pair $(s_1,s_2) \in \sbsat F \tau (s)$ such that $(x,y) \in \sbsat F \tau (s_1)$ or $(x,y) \in \sbsat F \tau (s_2)$. Suppose that $(x,y) \in \sbsat F \tau (s_1)$, then we are in one of the four following cases:
	\begin{enumerate}
		\item 	
		$
		\begin{tikzcd}
		s \ar[r,"\tau^*"] & s_1 \ar[r,"\tau"] \ar[dr,"\tau^*"] & s_2 \\
		&										& x \ar[r,"\tau"] & y
		\end{tikzcd}
		$
		\item 
		$
		\begin{tikzcd}
		s \ar[r,"\tau^*"] & s_1 \ar[r,equal] \ar[dr,"\tau^*"] & s_2 \\
		&									& x \ar[r,"\tau"] & y
		\end{tikzcd}
		$
		\item 
		$
		\begin{tikzcd}
		s \ar[r,"\tau^*"] & s_1 \ar[r,"\tau"] \ar[dr,"\tau^*"] & s_2 \\
		&									   & x \ar[r,equal] & y
		\end{tikzcd}
		$
		\item
		$
		\begin{tikzcd}
		s \ar[r,"\tau^*"] & s_1 \ar[r,equal] \ar[dr,"\tau^*"] & s_2 \\
		& 								  & x \ar[r,equal] & y
		\end{tikzcd}
		$	
	\end{enumerate}
	In every case, we can conclude that $(x,y) \in \sbsat F \tau (s)$. Thus $\sbsat F(\tau)$ is a left-semi-monoid.
	
	Finally, we show that $\sbsat F (\tau) \cdot \sbsat F (a) \le \sbsat F(a)$ for all $a \in A$. Let $s \in S$ and consider $(x,y) \in (\sbsat F \tau \cdot \sbsat F a)(s)$. Then $(x,y) \in \sbsat F a (s_1)$ or $(x,y) \in \sbsat F a (s_2)$ for some $(s_1,s_2) \in \sbsat F \tau (s)$. In the first case (and similarly for the second), it is
	\[
	\text{either } \quad
	\begin{tikzcd}
	s \ar[r,"\tau^*"] & s_1 \ar[r,"\tau"] \ar[dr,"\tau^*"] & s_2 \\
	&										& x \ar[r,"a"] & y
	\end{tikzcd}
	\text{ or } \quad
	\begin{tikzcd}
	s \ar[r,"\tau^*"] & s_1 \ar[r,equal] \ar[dr,"\tau^*"] & s_2 \\
	&									& x \ar[r,"a"] & y
	\end{tikzcd}
	\]
	and in both cases we have $(x,y) \in \sbsat F a (s)$, as required.
\end{proof}

\begin{remark}
	It is not true, in general, that $\sbsat F a \cdot \sbsat F \tau \le \sbsat F a$.  Indeed, consider $s \in S$ and $(x,y) \in (\sbsat F a \cdot \sbsat F \tau)(s)=\bigcup_{(s_1,s_2) \in \sbsat F a (s)} (\sbsat F \tau (s_1) \cup \sbsat F \tau (s_2))$. Then the following is one of four possible scenarios:
	\[
	\begin{tikzcd}
	s \ar[r,"\tau^*"] & s_1 \ar[r,"a"] \ar[dr,"\tau^*"] & s_2 \\
	&									& x \ar[r,"\tau"] & y
	\end{tikzcd}
	\]
	where it is clear that $(x,y)\notin \sbsat F a (s)$.
\end{remark}

\section{The category {\Meas}}

Our next goal is to discuss bisimulation for continuous Markov processes (see \cite{panangaden2009labelled,de_vink_bisimulation_1999}). In order to do this we need to step cautiously out of the world of sets and functions, and into that of measurable spaces and measurable functions. 

We recall that a measurable space $(X,\Sigma)$ is a set $X$ equipped with a $\sigma$-algebra, $\Sigma$, the algebra of measurable sets. A measurable function $f \colon (X,\Sigma_X) \longrightarrow(Y,\Sigma_Y)$ is a function $f\colon X \longrightarrow Y$ such that if $U$ is a measurable set of $(Y,\Sigma_Y)$, then $f^{-1} U$ is a measurable set of $(X,\Sigma_X)$. Together these form a category, {\Meas}.

\begin{lemma} 
{\Meas} has all finite limits and $\Gamma = \Meas(1,- ):\Meas\longrightarrow\Set$ preserves them. 
\end{lemma}

\begin{proof}
Let $F:D\longrightarrow\Meas$ be a functor from a finite category $D$. Then 
$\varprojlim F$ 
is the measurable space on the set $\varprojlim (\Gamma F)$ 
equipped with the least 
$\sigma$-algebra making the projections $\varprojlim (F) \longrightarrow F d$ measurable. 
\end{proof}

\begin{lemma}
    {\Meas} has coequalisers. If 
    \begin{tikzcd}
    (X,\Sigma_X) \ar[r, "f", shift left] 
      \ar[r, "g" below, shift right] & (Y,\Sigma_Y)
    \end{tikzcd}
    is a pair of parallel measurable functions, then their coequaliser is 
    $E: (Y,\Sigma_Y)\longrightarrow (Y/{\sim},\overline\Sigma)$, where $\sim$ is the 
    equivalence relation on $Y$ generated by $fx\sim gx$, and 
    $\overline\Sigma$ is the largest $\sigma$-algebra on $Y/{\sim}$ making $Y\longrightarrow Y/{\sim}$ measurable, {i.e.} 
    $\overline\Sigma = \{ V \ |\  e^{-1} V \in \Sigma_Y \}$.
\end{lemma}

\begin{corollary}
A morphism $e:(Y,\Sigma_Y) \longrightarrow (Z,\Sigma_Z)$ in {\Meas} is a regular epi if and only if $\Gamma e$ is a surjection in {\Set}, and $U\in\Sigma_Z$ iff 
$e^{-1}U\in\Sigma_Y$. 
\end{corollary}

\begin{corollary}
Any morphism in {\Meas} factors essentially uniquely as a regular epi followed by a monomorphism. 
\end{corollary}

However, {\Meas} is not regular because the pullback of a regular epi is not necessarily regular, as exhibited by this counterexample: 

\begin{example}
Let $(Y,\Sigma_Y)$ be the measurable space on $Y=\{a_0,a_1,b_0,b_1\}$ with 
$\Sigma_Y$ generated by the sets $\{a_0,a_1\}$ and $\{b_0,b_1\}$. Let 
$(Z,\Sigma_Z)$ be the measurable space on $Z=\{a'_0,a'_1,b'\}$, where the 
only measurable sets are $\emptyset$ and $Z$. Let $e \colon Y\longrightarrow Z$ be given by 
$e(a_i)= a'_i$, and $e(b_i) = b'$. Then $e$ is a regular epi. Now let 
$(X,\Sigma_X)$ be the measurable space on 
$X=\{a'_0,a'_1\}$ where $\Sigma_X = \{\emptyset, X\}$, and let $\colon :X\longrightarrow Z$ 
be the inclusion of $X$ in $Z$. Then $i^{*}Y = \{a_0,a_1\}$ with 
$\sigma$-algebra generated by the singletons, but $i^{*}e$ is not regular 
epi because $(i^{*}e)^{-1}\{a'_0\}=\{a_0\}$ is measurable, 
but $\{a'_0\}$ is not. 
\end{example}

The consequence of this is that {\Meas} has all the apparatus to construct a relational calculus, but that calculus does not have all the properties we expect. Specifically it is not an allegory. Accordingly, when we want to construct logical relations on {\Meas}, we will take the measurable spaces as structures in {\Set} and use the constructs in {\Set}.

\section{Probabilistic bisimulation}

We follow the standard approach by defining a continuous Markov process to 
be a coalgebra for the Giry functor. For simplicity we will work with 
unlabelled processes. 

\begin{definition}[Giry monad]
    Let $(X,\Sigma_X)$ be a measurable space. The Giry functor, $\Pi$, 
    is defined as follows, 
    $\Giry (X,\Sigma_X) = (\Giry X, \Giry\Sigma_X)$: 
    \begin{itemize}
        \item $\Giry X$ is the set of sub-probability measures on $(X,\Sigma_X)$. 
        \item $\Giry\Sigma_X$ is the least $\sigma$-algebra on $\Giry X$ such that for every $U\in\Sigma_X$, $\lambda\pi. \pi(U)$ is 
        measurable. 
    \end{itemize}
    If $f \colon (X,\Sigma_X) \longrightarrow (Y,\Sigma_Y)$ is a measurable function, then 
    $\Giry f (\pi) = \lambda V\in\Sigma_Y. \pi (f^{-1} V)$. $\Giry$ forms part of a monad in which the unit maps a point $x$ to the Dirac measure for $x$, and the multiplication is defined by integration, ~\cite{giry_categorical_1982}.
\end{definition}

\begin{definition}[continuous Markov process]
    A {\em continuous Markov process} is a coalgebra in {\Meas} for the Giry functor, {i.e.} a continuous Markov process with state space 
    $(S,\Sigma_S)$ is a measurable function $F \colon (S,\Sigma_S)\longrightarrow \Giry (S,\Sigma_S)$. A {\em homomorphism of continuous Markov processes} is simply a homomorphism of coalgebras. 
\end{definition}

There are now two similar, but slightly different approaches to defining the notion of a probabilistic bisimulation. \cite{panangaden2009labelled} follows Larsen and Skou's original definition for the discrete case. This begins by enabling a state space reduction for a single process and generates a notion of bisimulation between processes as a by-product. The second is the standard notion of bisimulation of coalgebras, as described in  \cite{rutten_universal_2000}. 

We begin with Panangaden's extension of the original definition of 
Larsen and Skou,  \cite{panangaden2009labelled,larsen_bisimulation_1991}.

\begin{definition}[Strong probabilistic bisimulation]\label{def:strong-prob-bisim-1}
Suppose $F\colon S \longrightarrow \Giry S$ is a continuous Markov process, then an equivalence relation $R$ on $S$ is a {\em (strong probabilistic) bisimulation} if and only if whenever $sRs'$, then for all $R$-closed measurable sets $U\in\Sigma_S$, $F s U = F s' U$. 
\end{definition}

We note that the $R$-closed measurable sets are exactly those inducing the 
$\sigma$-algebra on $S/R$, and hence that this definition of equivalence 
corresponds to the ability to quotient the state space to give a continuous Markov process on $S/R$. 

\begin{lemma}
    An equivalence relation $R$ on $(X,\Sigma_X)$ is a strong probabilistic bisimulation relation if and only if when we equip $X/R$ with the largest $\sigma$-algebra such that $X\to X/R$ is measurable, $X/R$ carries the structure of a Giry coalgebra and the quotient is a coalgebra homomorphism in {\Meas}. 
\end{lemma}

This definition assumes that $R$ is total. However that is not essential. 
We could formulate it for relations that are symmetric and transitive, but 
not necessarily total (partial equivalence relations). In this case we 
have a correspondence with subquotients of the coalgebra. We do, however, have to be careful that the domain of $R$ is a well-defined sub-algebra. 

Panangaden goes on to define a bisimulation between two coalgebras. We 
simplify his definition as we do not consider specified initial states. 

Given a binary relation $R$ between $S$ and $T$, we extend $R$ to a binary relation on the single set $S+T$. In order to apply the previous definition, we will want the equivalence relation on $S+T$ generated by $R$. 

Now $(S+T)\times (S+T) = (S\times S) + (S\times T) + (T\times S) + (T\times T)$, and each of these components has a simple relation derived from $R$, specifically $R\Op R$, $R$, $\Op R$ and $\Op R R$.

\begin{definition}[z-closed]
$R \subseteq S\times T$ is {\em z-closed} iff $R\Op R R \subseteq R$, in other words, iff whenever $sRt \wedge s_1Rt \wedge s_1Rt_1$ then $sRt_1$.
\end{definition}

\begin{lemma}
    $R \subseteq S \times T$ is z-closed if and only if 
    $R^{\ast} = R\Op R + R + \Op R + \Op R R$ is transitive as a 
    relation on 
    $(S+T)\times (S+T)$. Since $R^{\ast}$ is clearly symmetric, 
    $R$ is z-closed iff $R^{\ast}$ is a partial equivalence relation.
\end{lemma}

Secondly, given continuous Markov processes $F$ on $S$ and $G$ on $T$ we can define 
their sum $F+G$ on $S+T$:
\[
(F+G) x U = \begin{cases}
Fx(U \cap S) & \text{if $x \in S$} \\
Gx(U \cap T) & \text{if $x \in T$}
\end{cases}
\]

We can now make a definition that seems to us to contain the essence of 
Panangaden's approach: 

\begin{definition}[strong probabilistic bisimulation between 
processes]\label{def:strong-prob-bisim-2}
   $R$ is a strong probabilistic bisimulation between the continuous 
   Markov processes $F$ on $S$ and $G$ on $T$ iff
   $R^{\ast} = R\Op R + R + \Op R + \Op R R$ is a strong probabilistic bisimulation as defined in Definition \ref{def:strong-prob-bisim-1}
   on the sum process $F+G$ on $S+T$.
\end{definition}

Note that any such relation will be z-closed. Given that $R^{\ast}$ must 
be total, it also induces an isomorphism between quotients of the
continuous Markov processes. 

This definition corresponds exactly to what we get by taking the obvious 
logical relations approach. 

\paragraph*{Logical relations of continuous Markov Processes.}

Given a measurable space $(S,\Sigma_S)$, we treat the 
$\sigma$-algebra $\Sigma_S$ as a subset 
of the function space $[S\to 2]$, and use the standard mechanisms of 
logical relations in {\Set} to extend a relation 
$R\subseteq S\times T$ between two measurable spaces to a relation 
$\RSigma$ between $\Sigma_S$ and $\Sigma_{T}$: $U\RSigma V$ if and 
only if $\forall s,t. sRt \implies (s\in U \iff t\in V)$.

\begin{lemma}\begin{enumerate}
    \item If $R$ is an equivalence relation then $U\RSigma V$ iff $U=V$ and is $R$-closed. 
    \item If $R$ is z-closed, then $U\RSigma V$ iff $U+V$ is an $R^{\ast}$-closed subset of $S+T$.
    \item If $R$ is the graph of a function $f\colon S\longrightarrow T$, then 
    $U\RSigma V$ iff $U=f^{-1}V$. 
\end{enumerate}
\end{lemma}

Unpacking the definition of the Giry functor, a Giry coalgebra 
structure on the measurable space $(S,\Sigma_S)$
has type $S\longrightarrow [\Sigma_S \to [0,1]]$, or equivalently 
$S \times \Sigma_S \longrightarrow [0,1]$, where for the 
purposes of defining logical relations we regard $\Sigma_S$ as a 
subset of $[S\to 2]$. We again apply the standard machinery to this. 

\begin{definition}[logical relation of continuous Markov processes]
If $R\subseteq S\times T$ is a relation between the state spaces of 
continuous Markov processes $F\colon S\longrightarrow\Giry S$ and $G\colon T\longrightarrow\Giry T$, 
then $R$ is a {\em logical relation of continuous Markov processes} iff 
whenever $sRt$ and $U\RSigma V$, $F s U = G t V$. 
\end{definition}

The following lemmas follow readily from the definitions. 

\begin{lemma}
    If $R\subseteq S\times T$ is a total and onto z-closed relation between continuous Markov processes $F\colon S\longrightarrow\Giry S$ and 
    $G\colon T\longrightarrow\Giry T$, then $R$ is a logical relation of continuous Markov processes if and only if $R$ is a strong probabilistic 
    bisimulation. 
\end{lemma}

\begin{lemma}
     If $R\subseteq S\times T$ is the graph of a measurable function $f$ between continuous Markov processes $F\colon S\longrightarrow\Giry S$ and 
    $G\colon T\longrightarrow\Giry T$, then $R$ is a logical relation of continuous Markov processes if and only if $f$ is a homomorphism of continuous Markov processes. 
\end{lemma}

\begin{proof}
    Observe that $f$ is a homomorphism if and only if for all $s\in S$ and 
    $V\in\Sigma_{T}$, $G (fs) V = F s (f^{-1}V)$.
\end{proof}
  
So logical relations capture both the concept of strong probabilistic 
bisimulation (given that the candidate relations are restricted in 
nature), and the concept of homomorphism of systems. But they do not 
capture everything. 

\paragraph*{{$\Giry$}-bisimulation.}


Recall from~\cite{rutten_universal_2000} that for a functor $H \colon \sf C \longrightarrow \sf C$ and two $H$-coalgebras $f \colon A \longrightarrow HA$ and $g \colon B \longrightarrow HB$, an \emph{$H$-bisimulation} between $f$ and $g$ is a $H$-coalgebra $h \colon C \longrightarrow HC$ together with two coalgebra-homomorphisms $l \colon C \longrightarrow A$ and $r \colon C \longrightarrow B$, that is, it is a span in the category of coalgebras for $H$:
\[
\begin{tikzcd}
A \ar[d,"f"'] & \ar[l,"l"'] C \ar[r,"r"] \ar[d,"h"] & B \ar[d,"g"] \\
HA & \ar[l,"Hl"'] HC \ar[r,"Hr"] & HB
\end{tikzcd}
\]
where the above diagram is required to be commutative.

\begin{definition}
    A {$\Giry$}-bisimulation is simply an $H$-bisimulation in the category {\Meas} where the functor $H$ is $\Giry$. 
\end{definition}

It is implicit in this definition that a bisimulation includes a coalgebra structure, and is not simply a relation. Where the functor $H$ corresponds to a traditional algebra generated by first-order terms and equations, the algebraic structure on the relation is unique. But that is not the case here. 

\begin{example}\label{example:algebra-not-unique}
    Consider a continuous Markov process 
    $F: { S}\longrightarrow \Giry { S}$, then ${S}\times{ S}$ typically carries a number of continuous Markov process structures for 
    which both projections are homomorphisms. For example: 
    \begin{enumerate}
        \item a ``two independent copies'' structure given by: 
        \[FF (s,s') (U,U') = (F s U) \times (F s' U')\]
        \item a ``two observations of a single copy'' structure 
        given by: 
        \[
        F^2 (s,s') (U,U') = \begin{cases}
        F s (U\cap U') & \text{if $s = s'$} \\
        F s U \times F s' U' & \text{ if $s \ne s'$}
        \end{cases}
        \]
    \end{enumerate}
\end{example}
\begin{example}
More specifically, consider the process $t$ modelling a single toss of a 
fair coin. This can be modelled as a process with three states, 
$C=\{S,H,T\}$: Start (S), Head tossed (H) and Tail tossed (T). From S we 
move randomly to one of H and T and then stay there. The transition matrix 
is given below. This is a discrete process, and we take all subsets to be 
measurable. 
\[t\mbox{ is given by} \quad 
\begin{array}{c|ccc}
 & S & H & T \\
 \hline
 S & 0 & 0.5 & 0.5 \\
 H & 0 & 1 & 0\\
 T & 0 & 0 & 1
 \end{array}
 \]
Now consider the state space $C\times C$. We define two different process 
structures on this. The first, $\tdouble$, is simply the product of the two 
copies of $C$. The transition matrix for this is the tensor of the 
transition matrix for $C$ with itself: the pairwise product of the 
entries. This represents the process of two independent tosses of a coin. 

\[\tdouble\mbox{ is given by }  \quad
\begin{array}{c|*{9}{c}}
 & SS & HH & TT & HT & TH & SH & HS & ST & TS \\
 \hline
 SS & 0 & 0.25 & 0.25 & 0.25 & 0.25 & 0 & 0 & 0 & 0 \\
 HH & 0 & 1 & 0 & 0 & 0 & 0 & 0 & 0 & 0\\
 TT & 0 & 0 & 1 & 0 & 0 & 0 & 0 & 0 & 0\\
 HT & 0 & 0 & 0 & 1 & 0 & 0 & 0 & 0 & 0\\
 \ldots\\
 TS & 0 & 0 & 0.5 & 0 & 0.5 & 0 & 0 & 0 & 0
 \end{array}
 \]
The second, $\tsingle$ is identical except for the first row: 
\[
\tsingle\mbox{ is given by } \quad
\begin{array}{c|*{9}{c}}
 & SS & HH & TT & HT & TH & SH & HS & ST & TS \\
 \hline
 SS & 0 & 0.5 & 0.5 & 0 & 0 & 0 & 0 & 0 & 0 \\
 \ldots
 \end{array}
 \]
This is motivated by the process of two observers watching a single toss 
of a coin. 

The projections are homomorphisms for both these structures. For example, 
the first projection is a homomorphism for $\tsingle$ because for each 
$I$, $J$, $K$: 
\[ t I \{K\} = \sum_L \tsingle IJ \{KL\} \]
\end{example}

This means that in order to establish that a relation is a 
$\Giry$-bisimulation, we have to define a structure and prove the 
homomorphisms, and not simply validate some closure conditions. 

Moreover, in contrast to the case for first-order theories, this 
non-uniqueness of algebra structures implies that we can not 
always reduce spans of homomorphisms to relations. 

\begin{example}
Consider the sum of the two algebra structures from Example 
\ref{example:algebra-not-unique} as an algebra $\tdouble + \tsingle$
on 
$(C\times C)+(C\times C)$. This is a $\Giry$-bisimulation from $C$ to 
itself in which the legs of the span are the co-diagonal, $\nabla$,
followed by the 
projections. The co-diagonal maps $(C\times C)+(C\times C)$ to its 
relational image, but is not an algebra homomorphism for any algebra 
structure on 
$C\times C$. If there were an algebra homomorphism, for an algebra 
structure $\delta$, say, then we would have that both 
$(\tdouble+\tsingle)(\inl SS) (\nabla^{-1}\{HT\}) = \tdouble (SS) \{HT\}$ 
and 
$(\tdouble+\tsingle)(\inr SS) (\nabla^{-1}\{HT\}) = \tsingle (SS) \{HT\}$
would be equal to 
$\delta (SS) \{HT\}$. But the first is $\tdouble (SS) \{HT\} = 0.25$, 
and the second is $\tsingle (SS) \{HT\} = 0$.
\end{example}

We now show that, despite these issues, $\Giry$-bisimulations give rise to 
logical relations. 

\begin{theorem1}\label{thm:pi bisimulation implies logical relation} 
Suppose
\[
\begin{tikzcd}
S \ar[d,"F"'] & \ar[l,"l"'] P \ar[r,"r"] \ar[d,"H"] & T \ar[d,"G"] \\
\Giry S & \ar[l,"\Giry l"'] \Giry P \ar[r,"\Giry r"] & \Giry T
\end{tikzcd}
\]
is a $\Giry$-bisimulation between the continuous Markov processes 
$F$ and $G$. Let $R \subseteq S \times T$ be the relation which is 
the image of $\langle l,r\rangle \colon P \longrightarrow S\times T$, {i.e.} $sRt$ iff
$\exists p. lp=s \wedge rp=t$. 
Then $R$ is a logical relation between $F$ and $ G$.
\end{theorem1}
\begin{proof}
Suppose $sRt$ and $U\RSigma V$ for $U\in\Sigma_S$ and $V\in\Sigma_T$. 
We must show that $F(s)(U) = G(t)(V)$.

We begin by showing that $l^{-1}U = r^{-1}V$. Suppose $p\in P$, then 
$(lp)R(rp)$, and hence $p\in l^{-1} U$ iff $lp\in U$ iff $rp\in V$ (since 
$U\RSigma V$) iff $p\in r^{-1}V$. 
Hence $l^{-1}U = r^{-1}V$, as required. 

Now, since $sRt$, there is a $p$ such that $lp=s$ and $rp=t$. Then 
\begin{align*}
	F(s)(U) 
	&= H p (l^{-1} U) &\text{because $l$ is a $\Giry$-homomorphism}\\
	&= H p (r^{-1} V) &\text{because $l^{-1}U = r^{-1}V$}\\
	&= G(t)(V)        &\text{because $r$ is a $\Giry$-homomorphism} 
\end{align*}
as required. 
\end{proof}

Establishing a converse is more problematic. There are a number of issues.
One is that $\Giry$-bisimulations work on spans, not relations.
Another is that there might not 
be much coherence between the relation $R$ and the $\sigma$-algebras 
$\Sigma_S$ and $\Sigma_T$. And a third is the fact that in order to define 
a $\Giry$-algebra structure $H$ on $R$, we have to define $H (s,t) W$, 
where $W$ is an element of the $\sigma$-algebra generated by the sets 
$R\cap (U\times V)$, where $U\in\Sigma_S$ and $V\in\Sigma_T$. It is not 
clear that such an extension will always exist, and Example 
\ref{example:algebra-not-unique} shows that there is no canonical way to 
construct it. 

Nevertheless we can show that a logical relation gives rise to a
$\Giry$-bisimulation, unfortunately not on the original algebras, but on 
others with the same state space but a cruder measure structure. 

The following lemma is immediate. 

\begin{lemma}
Suppose $F \colon (S,\Sigma_S) \longrightarrow \Giry (S,\Sigma_S)$ is a continuous Markov 
process. Suppose also that $\Sigma'$ is a sub-$\sigma$-algebra of $\Sigma_S$, then $F$ restricts to a continuous Markov process $F'$ on 
$(S,\Sigma')$, and $1_S \colon (S,\Sigma_S) \longrightarrow (S,\Sigma')$ is a homomorphism. 
\end{lemma}

If $R$ is a logical relation between continuous Markov processes $F$ on 
$S$ and $G$ on $T$, then $R$ only gives us information about the 
measurable sets included in $\RSigma$. The following lemmas are immediate 
from the definitions. 

\begin{lemma}
    If $R\subseteq S\times T$ is a relation between the state spaces of 
    two continuous Markov processes $F$ and $G$ and $\proj 1 \colon R \longrightarrow S$, $\proj 2 \colon R \longrightarrow T$ are the two projections, then the following are 
    equivalent for $U\subseteq S$ and $V\subseteq T$:
    \begin{enumerate}
        \item $U [R\to\{0,1\}] V$ 
        \item $U$ is closed under $R\Op R$, and 
        $UR = V\cap \mathop{\mbox{cod}} R$
        \item $\invproj 1 U = \invproj 2 V$.
    \end{enumerate}
\end{lemma}

\begin{lemma} If $R\subseteq S\times T$ is a relation between the 
state spaces of 
    two continuous Markov processes $F$ and $G$, then the sets linked by 
    $[R\to\{0,1\}]$ have the following closure properties: 
\begin{enumerate}
    \item If $U [R\to\{0,1\}] V$ then $\compl U \ [R\to\{0,1\}]\ \compl V$ 
    \item If for all $\alpha\in A$, $U_\alpha [R\to\{0,1\}] V_\alpha$ then $\Union_{\alpha\in A} U_\alpha\  [R\to\{0,1\}]\  \Union_{\alpha\in A} V_\alpha$.
\end{enumerate}
\end{lemma}

\begin{corollary}
The {\em measurable} subsets linked by $[R\to\{0,1\}]$ have the same closure properties and hence the following are $\sigma$-algebras: 
\begin{enumerate}
    \item $\SRC S = \{ U\in \Sigma_S | \exists V\in\Sigma_T. U\RSigma V\}$
    \item $\SRC T = \{ V\in \Sigma_T | \exists U\in\Sigma_S. U\RSigma V\}$
    \item $\begin{array}[t]{cl}
        \SigmaR & = \{ W\subseteq R | \exists U\in\Sigma_S, 
    V\in\Sigma_T.\  U\RSigma V \wedge W=\invproj 1 U \}\\
    &= \{ W\subseteq R | \exists U\in\Sigma_S, 
    V\in\Sigma_T.\  U\RSigma V \wedge W=\invproj 1 U = \invproj 2 V\}.
    \end{array}$
\end{enumerate}
\end{corollary}

\begin{theorem1}\label{thm:logical relation implies pi bisimulation}
    Suppose $R\subseteq S\times T$ is a relation between the 
state spaces of two continuous Markov processes $F$ and $G$. If $R$ is a 
logical relation then there is a $\Giry$-bisimulation: 
\[
\begin{tikzcd}
(S,\SRC S) \ar[d,"F"'] & \ar[l,"\proj 1 {}"'] (R, \SigmaR) \ar[r,"\proj 2 {}"] \ar[d,"H"] & (T, \SRC T) \ar[d,"G"] \\
\Giry (S, \SRC S) & \ar[l,"\Giry {\proj 1 {}}"'] \Giry (R, \SigmaR) \ar[r,"\Giry {\proj 2 {}}"] & \Giry (T, \SRC T)
\end{tikzcd}
\]
\end{theorem1}

\begin{proof}
 Suppose $(s,t)\in R$ and $W\in\SigmaR$, then we need to define 
 $H (s,t) W$. Suppose $U\in\SRC S$, $V\in\SRC T$, such that 
 $W=\invproj 1 U = \invproj 2 V$ and $U\RSigma V$. Then, since R is a 
 logical relation, $F(s)(U) = G(t)(V)$. 
 
 We claim that this is independent 
 of the choice of $U$ and $V$. 
 Suppose $U'\in\SRC S$, $V'\in\SRC T$, such that 
 $W=\invproj 1 U' = \invproj 2 V'$ and $U'\RSigma V'$. 
 Then $\invproj 1 U' = \invproj 2 V = W$, and hence $U'\RSigma V$, so 
 $F(s)(U')=G(t)(V)=F(s)(U)$. 
 
 We now define $H(s,t)(W) = F(s)(U)$. 
 
 We need to show that this is a $\Giry$-algebra structure. 
 
 First, we show that $H(s,t)$ is a sub-probability measure. We use a slightly non-standard characterisation of measures: 
 \begin{enumerate}
     \item Since $\emptyset\in\SRC S$, $H(s,t)\emptyset=F(s)\emptyset = 0$.
     \item For $W$, $W'$ in $\SigmaR$, let $U$ and $U'$ be in $\SRC S$ 
     such that 
     $\invproj 1 U = W$ and $\invproj 1 U' = W'$. Then, since $F(s)$ is a measure:  
     $F(s)(U) + F(s)(U') = F(s)(U\cup U') + F(s)(U\cap U')$. Now, since 
     $\invproj 1 {}$  preserves unions and intersections, 
     $H(s,t)(W) + H(s,t)(W') = H(s,t)(W\cup W') + H(s,t)(W\cap W')$.
     \item If $W_i$ is an increasing chain of elements of $\SigmaR$, then 
     let $U_i$ be an increasing chain of elements of $\SRC S$ such that 
     $\invproj 1 (U_i) = W_i$. Then $H(s,t)(\Union W_i) = F(s)(\Union U_i) = \lim F(s)(U_i) = \lim H(s,t)(U_i)$. 
 \end{enumerate}

To complete the proof it suffices to show that for each $W\in\SigmaR$, 
$H(-)(W)$ is a measurable function. Choose $U\in\SRC S$ and $V\in\SRC T$ 
such that $U\RSigma V$ and $W=\invproj 1 U = \invproj 2 V$. Now, given 
$q\in [0,1]$, let $U_q = \{s\in S \mid F(s)(U)\leq q\}$ and 
$V_q = \{t\in T \mid G(t)(V)\leq q\}$. Now suppose $sRt$, then, since $R$ is 
a logical relation, $F(s)(U)=G(t)(V)$, hence $s\in U_q$ iff $t\in V_q$. 
Therefore $U_q \RSigma V_q$. Moreover, $H(s,t)(W)=F(s)(U)=G(t)(V)$, and 
hence $H(s,t)(W)\leq q$ iff $s\in U_q$ iff $t\in V_q$. It follows that 
$\{(s,t) | H(s,t)(W)\leq q \}\in\SigmaR$, and hence that $H(-)(W)$ is 
measurable as required. 
\end{proof}

Putting this together we see that if we have a logical relation between 
$F$ and $G$, then we get the following diagram, in which the non-horizontal maps in the top section are identities on state spaces: 
\[
\begin{tikzcd}
(S,\Sigma_S) \ar[d,"F"'] \ar[dr] & & 
  \ar[ll,"\proj 1 {}"']   (R,\Sigma_S\times\Sigma_T\upharpoonright R)
  \ar [d] \ar[rr,"\proj 2 {}"]
  & & \ar[dl] (T,\Sigma_T) \ar[d,"G"] \\
\Giry (S,\Sigma_S) \ar[dr] & (S,\SRC S) \ar[d,"F"'] & \ar[l,"\proj 1 {}"'] (R, \SigmaR) \ar[r,"\proj 2 {}"] \ar[d,"H"] & (T, \SRC T) \ar[d,"G"] 
 & \ar[dl] \Giry{(T,\Sigma_T)}\\
& \Giry (S, \SRC S) & \ar[l,"\Giry {\proj 1 {}}"'] \Giry (R, \SigmaR) \ar[r,"\Giry {\proj 2 {}}"] & \Giry (T, \SRC T)
\end{tikzcd}
\]

We can view Theorem \ref{thm:logical relation implies pi bisimulation} 
as saying that we may be given too fine a measure structure on $S$ and $T$ 
for a logical relation to generate a $\Giry$-bisimulation, but we can 
always get a $\Giry$-bisimulation with a coarser structure. Just how 
coarse and how useful this structure might be depends on the logical 
relation and its relationship with the original $\sigma$-algebras on 
the state spaces. 

\begin{example}
\begin{itemize}
    \item In the contrived examples of \ref{example:algebra-not-unique}, we have taken the relation $R$ to be the whole of $C\times C$ and in effect used the algebra structure to restrict the effect of this. However, since $R=C\times C$, $\SRC C$ contains only the empty set and the whole of $C$. As a result, the continuous Markov process we get is not useful: the probability of evolving into the empty set is always 0, and the probability of evolving into something is always 1. 
    \item In the same examples we can restrict the state spaces for 
    $\tdouble$ and $\tsingle$. For $\tdouble$ we take 
    $\Rdouble = \{SS,HH,TT,HT,TH\}$, reflecting the states accessible 
    from $SS$. In this case 
    $\SRC C = \{ \emptyset, \{S\}, \{H,T\}, \{S,H,T\}\}$. 
    For $\tsingle$ we take 
    $\Rsingle = \{SS,HH,TT\}$, and $\SRC C$ contains all the subsets of $C$. 
\end{itemize}
\end{example}

\paragraph*{Financial Support} Edmund Robinson and Alessio Santamaria acknowledge the support of EPSRC grant EP/R006865/1, Interface Reasoning for Interactive Systems. Santamaria also acknowledges the support of the Ministero dell'Universit\`a e della Ricerca Scientifica 
of Italy under Grant No.\ 201784YSZ5, PRIN2017 -- ASPRA
(\emph{Analysis of Program Analyses}).

\paragraph*{Competing Interests} The authors declare none. 

\bibliographystyle{msclike}
\bibliography{lr}

\end{document}